\documentclass[sigplan, 10pt]{acmart}
\usepackage{graphicx}
\usepackage{tabularx}
\usepackage{balance}  
\usepackage[misc,geometry]{ifsym}

\usepackage{subfig}
\usepackage{wrapfig}
\usepackage{float}
\usepackage{multirow, graphicx}
\usepackage{algorithm}
\usepackage{algpseudocode} 
\usepackage{amsmath}

\usepackage{xspace}

\usepackage{url}

\usepackage{breakurl}

\usepackage{sty/widetext} 

\usepackage{stackrel}

\ifdefined\TTSTYLETARGETreport
\usepackage[
    backend=bibtex,
    maxbibnames=100,
    firstinits=true,
    bibstyle=numeric,
    citestyle=numeric
]{biblatex}
\fi

\usepackage{graphicx, multirow}

\usepackage{tabularx, multirow, rotating} 
\usepackage{subfig} 
\usepackage{algorithm} 
\usepackage{algpseudocode} 
\usepackage{ulem} 
\usepackage{color} 

\usepackage{listings} 
\lstdefinestyle{Oracle}{basicstyle=\ttfamily,
                        keywordstyle=\lstuppercase,
                        emphstyle=\itshape,
                        showstringspaces=true,
                        }
\makeatletter
\newcommand{\lstuppercase}{\uppercase\expandafter{\expandafter\lst@token
                           \expandafter{\the\lst@token}}}
\newcommand{\lstlowercase}{\lowercase\expandafter{\expandafter\lst@token
                           \expandafter{\the\lst@token}}}
\makeatother

\usepackage{tikz}

\newcommand{\ballnumber}[1]{\tikz[baseline=(myanchor.base)] \node[circle,fill=.,inner sep=1pt] (myanchor) {\color{-.}\bfseries\scriptsize #1};}

\newif\ifboldnumber

\algrenewcommand\alglinenumber[1]{%
  \footnotesize\ifboldnumber\bfseries\fi\global\boldnumberfalse#1:}

\newtheorem{theorem}{Theorem}[section]

\usepackage{todonotes}
\AtBeginDocument{%
  \providecommand\BibTeX{{%
    \normalfont B\kern-0.5em{\scshape i\kern-0.25em b}\kern-0.8em\TeX}}}

\setcopyright{none}
\begin{document}
\setcounter{page}{1}

\title{
Cost-Effective Data Feeds to Blockchains via Workload-Adaptive Data Replication
}

\author{Kai Li}
\email{kli111@syr.edu}
\affiliation{%
  \institution{Syracuse University}
  \city{Syracuse}
  \state{NY}
  \postcode{13244}
}

\author{Yuzhe Tang} 
\authornote{\Letter\ \xspace{}Corresponding author}
\email{ytang100@syr.edu}
\affiliation{%
  \institution{Syracuse University}
  \city{Syracuse}
  \state{NY}
  \postcode{13244}
}

\author{Jiaqi Chen}
\email{jchen217@syr.edu}
\affiliation{%
  \institution{Syracuse University}
  \city{Syracuse}
  \state{NY}
  \postcode{13244}
}

\author{Zhehu Yuan}
\authornote{Work is done when the author is an undergraduate student at Syracuse University.}
\email{zy2262@nyu.edu}
\affiliation{%
  \institution{New York University}
  \city{New York}
  \state{NY}
  \postcode{10012}
}

\author{Cheng Xu}
\email{chengxu@comp.hkbu.edu.hk}
\affiliation{%
  \institution{Hong Kong Baptist University}
  \city{Kowloon}
  \state{Hong Kong}
  \postcode{SAR}
}

\author{Jianliang Xu}
\email{xujl@hkbu.edu.hk}
\affiliation{%
  \institution{Hong Kong Baptist University}
  \city{Kowloon}
  \state{Hong Kong}
  \postcode{SAR}
}



\begin{abstract}
Feeding external data to a blockchain, a.k.a. data feed, is an essential task to enable blockchain interoperability and support emerging cross-domain applications, notably stablecoins. Given the data-intensive feeds in real life (e.g., high-frequency price updates) and the high cost in using blockchain, namely Gas, it is imperative to reduce the Gas cost of data feeds. Motivated by the constant-changing workloads in finance and other applications, this work focuses on designing a {\it dynamic, workload-aware} approach for cost effectiveness in Gas. This design space is understudied in the existing blockchain research which has so far focused on static data placement.

This work presents GRuB, a cost-effective data feed that dynamically replicates data between the blockchain and an off-chain cloud storage. GRuB's data replication is workload-adaptive by monitoring the current workload and making online decisions w.r.t. data replication. A series of online algorithms are proposed that achieve the bounded worst-case cost in blockchain's Gas. GRuB runs the decision-making components on the untrusted cloud off-chain for lower Gas costs, and employs a security protocol to authenticate the data transferred between the blockchain and cloud. The overall GRuB system can autonomously achieve low Gas costs with changing workloads. 

We built a GRuB prototype functional with Ethereum and Google LevelDB, and supported real applications in stablecoins. Under real workloads collected from the Ethereum contract-call history and mixed workloads of YCSB, we systematically evaluate GRuB's cost which shows a saving of Gas by $10\%\sim{}74\%$, with comparison to the baselines of static data-placement. 
\end{abstract}

\maketitle


\newcommand{\tremark}[1]{\footnote{\textcolor{red}{(Ting's comment: #1)}}}
\newcommand{\xremark}[1]{\footnote{\textcolor{red}{(Xin's comment: #1)}}}
\newcommand{\jj}[1]{\footnote{\textcolor{blue}{(Jiyong: #1)}}}
\newcommand{\yz}[1]{\footnote{\textcolor{red}{(Yuzhe: #1)}}}

\definecolor{mygreen}{rgb}{0,0.6,0}
\lstset{ %
  backgroundcolor=\color{white},   
  basicstyle=\scriptsize\ttfamily,        
  breakatwhitespace=false,         
  breaklines=true,                 
  captionpos=b,                    
  commentstyle=\color{mygreen},    
  deletekeywords={...},            
  escapeinside={\%*}{*)},          
  extendedchars=true,              
  keepspaces=true,                 
  keywordstyle=\color{blue},       
  language=Java,                 
  morekeywords={*,...},            
  numbers=left,                    
  numbersep=5pt,                   
  numberstyle=\scriptsize\color{black}, 
  rulecolor=\color{black},         
  showspaces=false,                
  showstringspaces=false,          
  showtabs=false,                  
  stepnumber=1,                    
  stringstyle=\color{mymauve},     
  tabsize=2,                       
  title=\lstname,                  
  moredelim=[is][\bf]{*}{*},
}

\providecommand{\lkvPut}{\textsc{Put}\xspace}
\providecommand{\lkvGet}{\textsc{Get}\xspace}
\providecommand{\lkvScan}{\textsc{ScanQ}\xspace}

\providecommand{\lkvPutTwo}{\textsc{bPut2}\xspace}
\providecommand{\lkvGetTwo}{\textsc{qGet2}\xspace}
\providecommand{\lkvScanTwo}{\textsc{qScan2}\xspace}

\section{Introduction}
A smart contract is a user program that runs on a blockchain, such as Ethereum~\cite{me:eth} and EOS.IO~\cite{me:eosio}. It holds the promises to expand the blockchain's functionalities from the basic cryptocurrency payments to broader applications in decentralized finance, supply chains, online gaming, et al. Feeding external data onto the blockchain, a.k.a. data feed, is an essential task to enable these blockchain applications. Today, data feeds are widely adopted, notably in decentralized finance. 
For instance, stablecoins, a cryptocurrency with stable price that sees an explosion of interest (as in Facebook's Libra~\cite{me:libra}) and deployment (as in the popular DAI~\cite{me:dai:etherscan} and Tether~\cite{me:tether:etherscan} tokens on Ethereum) since 2019, require feeding real-world asset prices to the blockchain, for instance the {\it Ether-price feed} used in DAI~\cite{me:dai:etherscan}. For another instance, to enable asset exchange across different blockchains, say allowing a Bitcoin owner to transact with an Ethereum token owner, it entails a {\it ``side-chain'' feed} such as BtcRelay~\cite{me:btcrelay,me:btcrelay:github,DBLP:conf/sp/ZamyatinHLPGK19} to send the recently found Bitcoin blocks onto Ethereum for verifying Bitcoin deposit.
There are many other blockchain applications that have been or can be enabled by data feeds, including decentralized insurance~\cite{DBLP:conf/ccs/ZhangCCJS16}, tracing supply-chains~\cite{Feng_Tian_2016,me:ibmmaersk}, healthcare~\cite{DBLP:conf/healthcom/Mettler16}, transparency logging~\cite{me:keybase,DBLP:conf/sp/TomescuD17,DBLP:conf/usenix/AliNSF16}, trustless information-security~\cite{me:keybase:bkc}, et al.

%
Operating today's data feeds can be an expensive business. Specifically, many real-world data feeds generate an intensive stream of data updates at a high frequency (e.g., the price updates in seconds and microseconds). Under these data-intensive streams, data feeds, if improperly designed, could cause a heavy use of blockchain and lead to high monetary cost, known as Gas~\cite{wood2014ethereum}. 
{\color{black}
The expense burdens not only data-feed operators (e.g., ChainLink and MakerDAO) but also the financial applications running on top of the data feeds (e.g., decentralized exchanges such as AmpleForth and Syntheix~\cite{DBLP:journals/corr/abs-2005-04377}), eventually transferring to high fees for end users (e.g., users of decentralized exchanges)}. It is thus imperative to design cost-effective data feeds for scaling blockchain applications to support real-world data-intensive scenarios.

The goal of this work is to explore how a {\it dynamic, workload-aware} design of data feed can effectively save Gas. The design goal is motivated by 1) the observation that real-world financial applications exhibit highly dynamic workload patterns, which present opportunities to reduce costs --- Intuitively, if one can dynamically adjust the location of the data feeds (w.r.t. the blockchain) according to the current data supply-demand, the Gas cost caused by the repeated use of blockchains could be avoided. See the next two paragraphs for a detailed justification. 2) 
Furthermore, the design space of a workload-aware approach has not been studied in the existing blockchain-systems research. While there is a large body of research works on reducing blockchain costs, notably the layer-two protocols exemplified by payment channels~\cite{me:lightning,poon2016bitcoin,DBLP:journals/corr/MillerBKM17,DBLP:journals/corr/abs-1804-05141} that aim to place application logic off the blockchain, all existing approaches are based on static data placement. That is, the placement of data and computation w.r.t. blockchains stays fixed once the system starts running, and it does not reflect the constant change in the workloads. The design space of a dynamic, workload-aware approach to optimize the smart-contract cost (in Gas) for data feeds and blockchain applications beyond is an uncharted territory. 

This work presents GRuB, a workload-adaptive data replication framework for cost-effective data feeding. 
The system model is a data pipeline involving three actors: As illustrated in the left part of Figure~\ref{fig:systemmodel:v4}, an off-chain {\it data producer (DO)} feeds a stream of data updates to multiple {\it data-consumer smart contracts (DUs)} on the blockchain. The data flow is coordinated by an intermediary{\it key-value (KV) store} between the DO and DUs. 
A conventional design of data feed is to realize the KV store in a smart contract that accepts DO's data updates in transactions and DU's queries in contract internal calls. An alternative design is to statically place the KV store off the blockchain (e.g., the static off-chain feed, TownCrier~\cite{DBLP:conf/ccs/ZhangCCJS16}). By contrast, GRuB is a KV store built on {\it hybrid} storage media: By default, the data updates are persisted on an off-chain cloud storage provider (SP) such as Amazon S3~\cite{me:amazons3} and upon DU's queries, are brought to the blockchain, buffered in a smart-contract memory. Optionally, the buffered data can be persisted to the smart-contract storage, as a data replica, to benefit future read queries. GRuB's system model is illustrated by the right part of Figure~\ref{fig:systemmodel:v4}.
{\color{black} Note that our system model assuming any cloud provider is untrusted should be differentiated from the multi-cloud model adopted in existing works~\cite{DBLP:journals/tos/BessaniCQAS13,DBLP:conf/cloud/Abu-LibdehPW10} which trusts at least one cloud provider.}

\begin{figure}[!bhtp]
\centering
\includegraphics[width=0.5\textwidth]{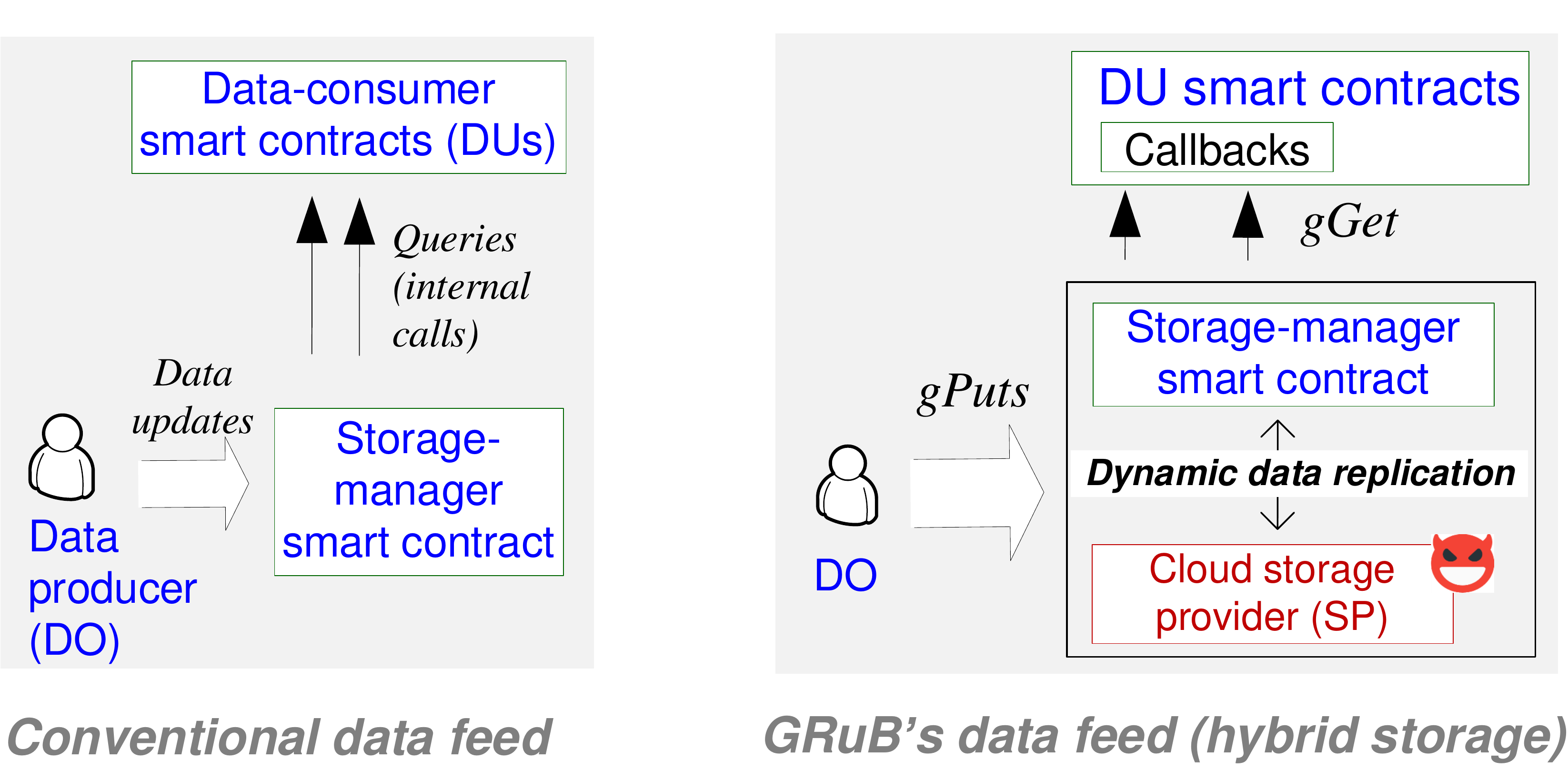}
\caption{
System model: The conventional data feed (left part) shows an on-chain KV store in a smart contract that mediate between an off-chain data producer (DO) and multiple on-chain data-consumer smart contracts (DUs). 
GRuB (the right part) introduces an off-chain cloud storage provider (SP) in the system model, which, together with a smart contract, provides a hybrid KV store that exposes a remote-procedure call interface (i.e., \texttt{gPuts}) to the DO and an internal-call interface (i.e., \texttt{gGet} with callbacks) to the DU smart contracts. With the default storage on SP, the fed data is dynamically replicated onto the blockchain by demand (in blue). 
Green in this figure illustrates smart contracts running on a blockchain and red is the SP who is the primary adversary in our trust model.}
\label{fig:systemmodel:v4}
\end{figure}

The key decision to make in GRuB is whether and when a data record in the feed should be replicated onto the smart-contract storage on a blockchain. 
{\it Always} storing a replica of the data being read, on the one hand, can benefit future data reads by avoiding loading data onto the blockchain repeatedly. On the other hand, if there are no future reads, such a data replica would be wasted. 
Thus, GRuB chooses to replicate data in a {\it workload-adaptive} manner: If the current workload is dominated by the reads from DUs, the GRuB would decide to store a data replica on the blockchain. Otherwise, if the current workload is dominated by the updates from the DO, the GRuB would decide to avoid replicating data on chain.
This design systematically avoids the two most expensive operations in Gas. That is, replicating data on chain under read-intensive workloads can avoid the expensive transactions otherwise needed to bring data onto the blockchain, and evicting data replicas under write-intensive workloads can prevent the expensive storage writes in smart contracts. 
See Section~\ref{sec:costmodel} for details on Ethereum's Gas-based cost model and Section~\ref{sec:design} for a basic measurement study that corroborates our insight here.

Dynamic decision making w.r.t. data replication has been a well-studied research topic in conventional distributed systems. Briefly, a common approach~\cite{DBLP:conf/icde/HuangW93} is to model the target system by multiple ``sites'', and run workload monitoring and decision making distributedly on each site. These solutions lay an important foundation for designing dynamic data-replication in GRuB.
However, simply using them as they are in GRuB is insufficient. 
Notably, existing dynamic replication frameworks are not designed with blockchain's Gas cost model in mind or do not reflect the GRuB's cost to enforce data security (e.g., on untrusted SP off-chain). If used improperly in GRuB, they may lead to excessive costs; for instance, the Gas model charges higher unit cost (e.g., per word) for ``local'' operations in smart contract (e.g., on-chain storage updates) than for data movement over the network (by transactions). Such a unique cost characteristic may invalidate the existing design that collocates the decision making with data replicas.  

To fill the gap, GRuB presents a Gas-aware data-replication system which places workload monitoring and decision making off the blockchain. We propose a security protocol to guarantee the integrity of workload trace and replication decisions that are transferred from untrusted off-chain SP to the blockchain. The decisions in GRuB are made by a Gas-aware online algorithm that achieves the bounded ``Gas competitiveness'' -- Specifically, the worst-case Gas caused by the data replication following the decision made by this online algorithm is bounded by a small-constant multiplicative factor (e.g., 2) to that caused by an optimal offline algorithm. 
{\color{black} 
This work emphasizes building a data-replication {\it mechanism} supporting sample policies to bound competitiveness. A comprehensive study of policies to configure the mechanism is out of the scope.} 
Overall, GRuB can autonomously run in the hybrid data feeds with changing workloads, while keeping the Gas low.

GRuB's system is generic: To support applications, GRuB exposes an extensible KV store interface (API) that supports Puts from the DO and Gets with callbacks to process queries in a DU contract. GRuB can be built relying on generic interfaces of the underlying systems (similar to an ABI); that is, any blockchain supporting smart contracts and any off-chain storage services supporting KV storage. We have built a GRuB prototype functional with Ethereum~\cite{me:eth} and Google LevelDB~\cite{me:leveldb}, and used it to enable a financial application: A stablecoin collateralized with an Ether-price feed built on GRuB.
Based on the real-world workloads collected from Ethereum, we evaluate GRuB's Gas cost, which shows that GRuB can save up to $67\%$ Gas compared to the static-data-placement baselines. For more extensive evaluation, we build a benchmark by mixing the YCSB workloads. The evaluation under YCSB benchmark shows that compared to the baselines, GRuB can save Gas by $10\%\sim{}74\%$ depending on specific record sizes and read-write ratios.

The contributions of this paper are outline as following: 

1. Propose a dynamic, workload-adaptive approach by mixing on-chain and off-chain data storage to optimize the smart-contract costs. To the best of our knowledge, this identifies an unexplored design space in the existing blockchain research. 

2. Present GRuB, a Gas-efficient data feed by dynamically replicating data between the hybrid data storage on and off the blockchain. GRuB employs new techniques, a Gas-aware online algorithm for replication decision-making and a security-centric protocol for running the decision components off-chain at a low cost. 

3. Validate the applicability of GRuB and evaluate its cost in Gas extensively, by systematically studying real-world applications, building a benchmark suite from real-world traces, and evaluating the costs. The result shows that GRuB can achieve a Gas saving by $10\%\sim{}74\%$ when compared to static data-placement baselines. 

\section{Design Motivations}

\subsection{Preliminary on Motivating Applications}
\label{sec:motivateapps}

Data feeding enables a blockchain to be able to interoperate with external worlds (i.e., the blockchain interoperability), which further enables a good number of deployed blockchain applications in cross-domain scenarios. Here, we describe two such applications in detail, as an effort to motivate our work.

{\it Stablecoins (on price feeds)}: Unlike Bitcoin, Ether and other ``native'' cryptocurrencies, a stablecoin is a cryptocurrency with stable prices. 
Price stability is the key requirement for real-world adoption of today's cryptocurrencies in realistic applications (e.g., loans, derivatives, and prediction markets). Recently, there is an explosion of stablecoins proposed (e.g., Facebook Libra) and deployed (e.g., DAI~\cite{me:dai:etherscan}, Tether~\cite{me:tether:etherscan}, and the other 57 stablecoins operational on Ethereum, as of May 2020~\cite{me:stablecoin}). 

There are different approaches to realize price stability~\cite{10.1145/3387945.3388781}: A stablecoin can be either directly backed by a stable asset (e.g., USD or gold) or indirectly backed via yet another cryptocurrency. The latter design, named indirectly-backed stablecoin, has the benefit of not relying on a trusted third-party vault off-chain to keep collateral and is adopted in popular stablecoins such as DAI~\cite{me:dai:etherscan} which is indirectly backed by Ether. To manage the price instability of Ether itself, the DAI runs a smart contract on Ethereum that controls the issuance and redemption of DAI. To make each DAI redeemable with one-USD worth of Ether, the DAI smart contract needs to be aware of the current price of Ether (or Ether-USD exchange rate). This is done by a price feed in practice~\cite{me:pricefeed:makerdao}, which upload the stream of price updates from a trusted source off-chain, such as Coinbase.\footnote{The off-chain party trusted by an indirectly-backed stablecoin performs a much simpler task than that by the directly-backed stablecoin. The former is a price feed, while the latter is a full-fledged vault storing the collateralized asset, subject to the public auditing~\cite{10.1145/3387945.3388781,moin2019classification}.} 

{\it Cross-chain swaps (on side-chain feeds)}: Supporting asset swaps across multiple blockchains is an important financial application paradigm, enabling asset liquidity on Blockchains. 
For instance, there are Bitcoin-pegged ERC20 tokens on Ethereum~\cite{me:btcpegged:etherscan} which allow a Bitcoin owner to transact with an asset owner on Ethereum. 
An efficient approach to enable such applications is the {\it side-chain paradigm} where blockchain A feeds its produced blocks to smart contracts running on blockchain B. For instance, BtcRelay~\cite{me:btcrelay,me:btcrelay:github} is such a side-chain feed connecting Bitcoin and Ethereum. BtcRelay style side-chain feeds are widely used in Bitcoin-pegged ERC20 tokens (e.g., tBTC~\cite{me:tbtc:github,me:tbtc:website,me:tbtc:etherscan} and others~\cite{me:btcpegged:etherscan,DBLP:conf/sp/ZamyatinHLPGK19}), 
Ethereum lottery games~\cite{me:ethlottery:reddit,me:ethlottery:etherscan}, et al. 

Other than the above two classes of data feeds, there are many other uses of data feeds, either deployed or envisioned. For instance, running flight insurances on Ethereum requires data feeds to provide flight cancel/deploy information.  Running stock exchanges may require an off-chain order book to feed stock/order prices. In other domains, blockchains are envisioned to support the auditing of transparency logs~\cite{me:keybase,DBLP:conf/sp/TomescuD17}, where the smart contracts running auditing logic need data feeds of log updates from off-chain servers. 

\begin{figure}
\centering
  \includegraphics[width=0.325\textwidth]{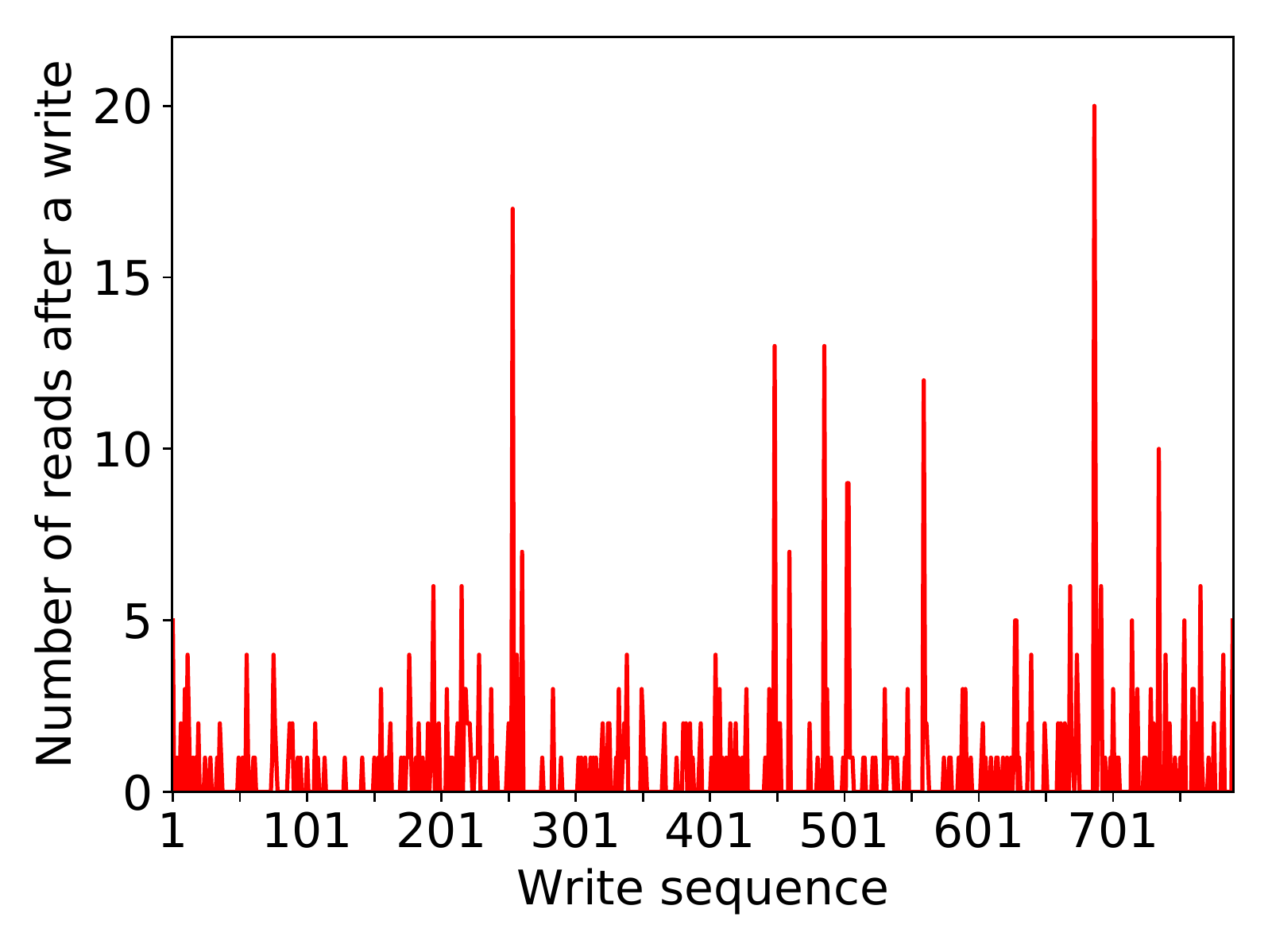}%
\caption{The workloads of ethPriceOracle~\cite{me:ethpriceoracle} that feed the MakerDAO stablecoin platform~\cite{me:makerdao} on Ethereum}
  \label{fig:workload:ethpriceoracle:1}%
\label{fig:workload:ethpriceoracle}
\end{figure} 

\begin{table}[!htbp] 
\caption{Distribution of writes by the number of reads followed in the ethPriceOracle trace (\#r represents the number of reads per write)}
\label{tab:ethpricefeed:distribution}\centering{\small
\begin{tabularx}{0.4\textwidth}{ |X|c|c|c|c|c| }
  \hline
\#r & Percentage & \#r & Percentage & \#r & Percentage \\ \hline
0 & $70.4\%$ & 5 & $0.76\%$ & 10 & $0.13\%$\\ \hline
1 & $16.0\%$ & 6 & $0.63\%$ & 12 & $0.13\%$\\ \hline
2 & $6.46\%$ & 7 & $0.25\%$ & 13 & $0.25\%$\\ \hline
3 & $2.91\%$ & 8 & $0.13\%$ & 17 & $0.13\%$ \\ \hline
4 & $1.52\%$ & 9 & $0.25\%$ & 20 & $0.13\%$ \\ \hline
\end{tabularx}
}
\end{table}

{\bf Workloads}: In these applications, the workload a data-feed serves consists of data reads from the consumer smart-contracts and the updates from the data producer. One of the motivating observations of this work is that many real-world workloads in data feeds fluctuate widely in the read-write ratio. Here, we present a measurement result as an example. 
EthPriceOracle~\cite{me:ethpriceoracle} is a price feed operational in the Ethereum mainnet and in use to support indirectly-backed stablecoin DAI, as part of the MakerDAO platform~\cite{me:makerdao}. EthPriceOracle allows 14 off-chain accounts to update the price feed and is implemented as a smart contract supporting a price-update function (i.e., \texttt{poke()}) and a price-read function (i.e., \texttt{peek()}).
We collected a call trace of \texttt{poke()} and \texttt{peek()} between April 25th, 2018 to April 30th, 2018; the collection is done in two means, by running an Ethereum full node and by querying a public Ethereum dataset hosted on Google BigQuery~\cite{me:ethtrace:bigquery}. Figure~\ref{fig:workload:ethpriceoracle} plots the 5-day trace where each X tick is a data-feed update (i.e., a \texttt{poke()} call) and the Y value associated with a X value is the number of data-feed reads (i.e., \texttt{peek()} calls) immediately following the write in the call trace. 
The workload distribution is also summarized in Table~\ref{tab:ethpricefeed:distribution}.
It can be seen that the number of reads following a write fluctuate; half of Y values are 0 and 1, but occasionally it also reaches as high as 20 reads after a write.

While this is the case of one particular application, the data-feed workloads being fluctuating commonly apply. Because in a typical data feed, the updates are produced continuously  at a regular rate, while the reads from the data consumer smart contract are by demand, which typically come and go in an ad-hoc fashion. 

\subsection{System Model and Trust Model}
\label{sec:systemmodel}

In this subsection, we formally describe the system model introduced before.
Recall Figure~\ref{fig:systemmodel:v4} that our system model includes three parties: A data producers (DO), a key-value (KV) store (i.e., the GRuB) and a number of data-consumer smart contracts (DUs). The off-chain DO sends data updates to the KV store, by invoking its function, \texttt{gPuts}. A DU smart contract queries the data feed stored in the KV store by issuing a function call to \texttt{gGet}.
The two functions exposed by the KV store are described by Listing~\ref{lst:grubapi}.

\definecolor{mygreen}{rgb}{0,0.6,0}
\lstset{ %
  backgroundcolor=\color{white},   
  basicstyle=\small\ttfamily,        
  breakatwhitespace=false,         
  breaklines=true,                 
  captionpos=b,                    
  commentstyle=\color{mygreen},    
  deletekeywords={...},            
  escapeinside={\%*}{*)},          
  extendedchars=true,              
  keepspaces=true,                 
  keywordstyle=\color{blue},       
  language=Java,                 
  morekeywords={*,...},            
  numbers=none,                    
  numbersep=5pt,                   
  numberstyle=\small\color{black}, 
  rulecolor=\color{black},         
  showspaces=false,                
  showstringspaces=false,          
  showtabs=false,                  
  stepnumber=1,                    
  stringstyle=\color{mymauve},     
  tabsize=2,                       
  caption = {GRuB APIs},
  label = {lst:grubapi},
  moredelim=[is][\bf]{*}{*},
}
\begin{lstlisting}
//external call by off-chain DO
bool gPuts(KV[] kvs); 
//internal call by smart contract (DU)
KV[] gGet(Key k1, Callback cb);
\end{lstlisting}

Specifically, a single \texttt{gPuts} call by the data producer batches multiple KV records in an epoch (e.g., every 1 min.) to update the KV store.
A \texttt{gGet} call issued by a DU smart contract retrieves KV records by a specified data key and returns its control to an optional callback function in the caller smart contract. The callback function often executes query-processing logic based on the retrieved KV records. Here, note that the caller of \texttt{gPuts} is the off-chain data producer and it can be implemented as a remote-procedure call, for instance, in Python.
The caller of \texttt{gGet} is a smart contract and it can be implemented as a Solidity function.

GRuB is a KV store based on ``hybrid'' storage media both on and off the blockchain. On the blockchain, it runs a storage-manager smart contract. Off the blockchain, it runs a KV store instance on an untrusted cloud storage provider (SP), such as Amazon S3.

GRuB can be used as a base to support different domain applications. To do so, an application developer writes a DU smart contract encoding the application logic and embedding a query-processor function to be called by \texttt{gGet}. GRuB can enable a price feed: Recall Section~\ref{sec:motivateapps} that a price feed supports a price-update function \texttt{poke()} and a price-read function \texttt{peek()}. These two functions can be mapped to GRuB's \texttt{gPuts} and \texttt{gGet}, respectively, by modeling the price of each collateral asset as a KV record (e.g., $\langle{}\text{Ether},150\text{USD}\rangle{}$). Section~\ref{sec:cases} presents two end-to-end applications built on GRuB.

\label{sec:trustmodel}
{\bf Trust model}: In our system, the primary adversary is the untrusted cloud storage provider who can forge, replay, omit and fork the data sent to the blockchain, in order to break the data integrity. The ``data'' includes the KV records, proofs and various protocol-specific metadata including collected trace of workloads and replication decisions.
We assume high availability among all participating parties and exclude denial-of-service attacks from the scope of this paper. 
All smart contracts including the application smart contracts and GRuB's storage-manager contracts are trusted in terms of program security (no exploitable security bugs), execution non-stoppability, etc. We also make standard assumption on blockchain security that the blockchain is immutable, fork-consistent and Sybil-secure. The underlying security assumption is that a deployed blockchain system runs among a large number of peers where majority of them are honest peers and compromising the majority is hard.

\begin{table}[!htbp] 
\caption{Ethereum's Gas cost w.r.t. different operations~\cite{wood2014ethereum}: Operations related to data movement (transactions) and storage updates are the most expensive in Gas.}
\label{tab:costmodel}\centering{\scriptsize
\begin{tabularx}{0.45\textwidth}{ |X|l| }
  \hline
 Operation & Gas cost ($X$ is the number of 32-byte words) \\ \hline
 Transaction & $C_{tx}(X) = 21000+2176X$ ($X<1000$) \\ \hline
 Storage write (insert) & $C_{insert}(X)=20000X$ \\ \hline
 Storage write (update) & $C_{update}(X)=5000X$ \\ \hline
 Storage read & $C_{read}(X) = 200X$ \\ \hline
 Hash computation & $C_{hash}(X) = 30+6X$ \\ \hline
\end{tabularx}
}
\end{table}

\label{sec:costmodel}
{\bf Cost model}: The primary cost considered in this work is the cost in using blockchains and executing smart contracts. This paper considers the use of Ethereum. Table~\ref{tab:costmodel} presents the Ethereum cost model in Gas (the cost unit in Ethereum). It can be seen the most expensive operations in Gas per word are transactions and storage writes/updates. 
In our system model, the use of cloud service (SP) may also lead to expenses, which however is much cheaper than that of blockchains: Consider storing one gigabytes in today's cloud storage, which falls under the free tier for all major providers (i.e., Amazon S3, Dropbox, et al), leading to zero-dollar spending, whereas doing the same on Ethereum costs more than $\$231$ million USD (with the Ether price as of Nov. 2019). Because of this, the cloud-service fee in our target application is negligible compared with the Gas cost from blockchains.

Also reducing the Gas of a blockchain application implies improving the throughput of this application, because 1) the transaction throughput of a blockchain is bounded by the total Gas a block can take, such as 10 million gas per Ethereum block; reducing the Gas per operation implies the application can submit more operations in a given time. 2) We assume blockchain is the bottleneck of a target application, which currently takes tens of transactions per second and is much lower than that of conventional computer systems, even for a single machine. Thus, the main goal of this work is to reduce the Gas cost of a blockchain application. 

\subsection{Motivating Cost Observation}
\label{sec:design}

\begin{figure}[!ht]
  \begin{center}
  \includegraphics[width=0.35\textwidth]{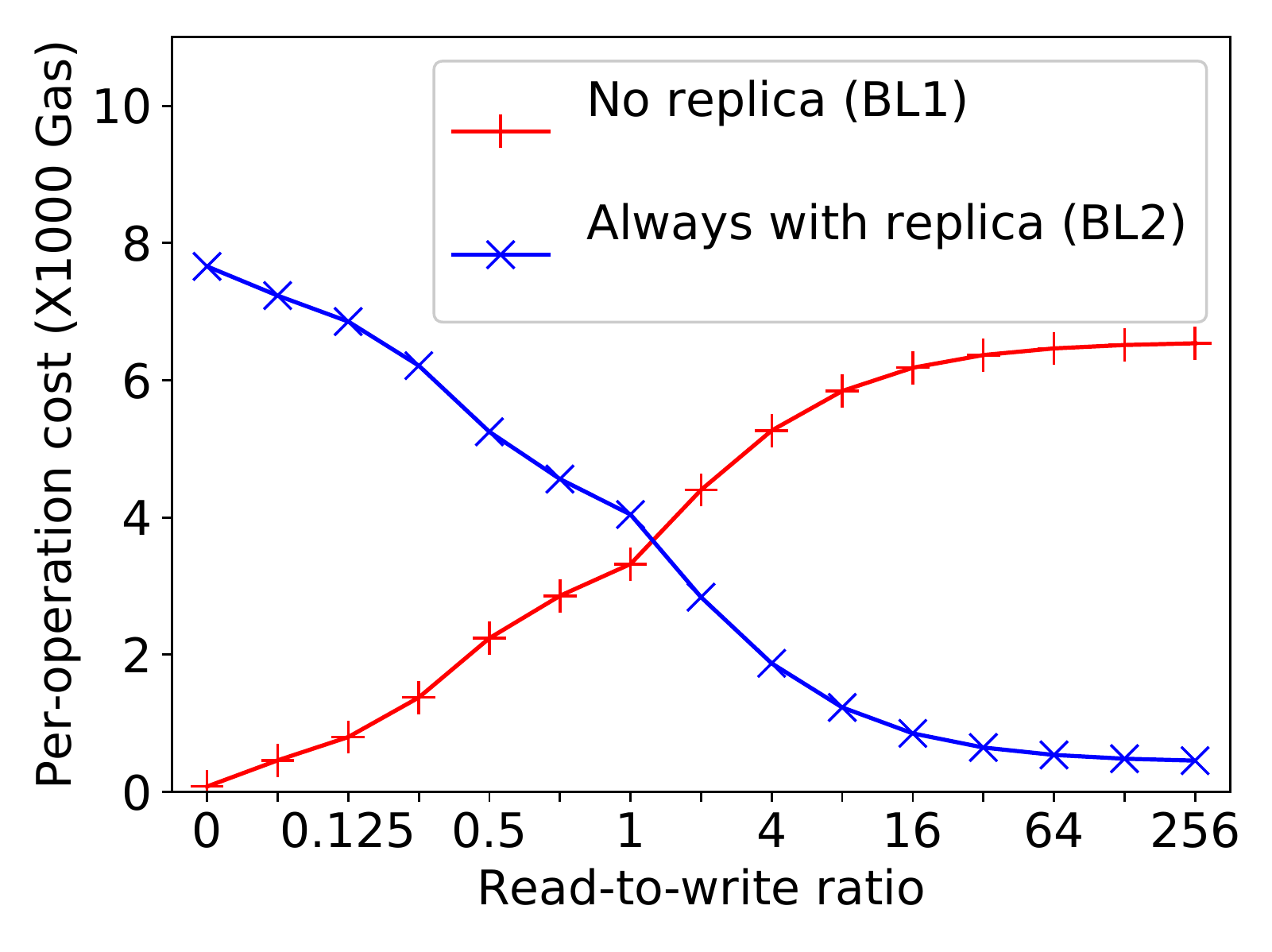}%
  \end{center}
\caption{Preliminary Gas measurements of static baselines}
  \label{fig:prel:mhtdataset}%
\end{figure}

{\bf 
Design Space}: This work addresses the design of hybridized data storage over blockchain and SP. We consider the two design baselines: 1) data is only stored on the off-chain SP and is brought into the smart-contract memory when serving \texttt{gGet}. This baseline is named BL1. Alternatively, 2) data is stored both on the off-chain SP and on blockchain. The baseline is named BL2. Note that our cost model only considers blockchain-induced cost, Gas, and excludes the off-chain costs including cloud service fee (on SP). Thus, BL2's cost is equivalent  to the design of placing data storage only on the blockchain. Note that these two baselines are based on {\it static} decisions regarding data replication.

{\bf Measurement observation}: To motivate dynamic data replication of this work, we conduct a rapid measurement study: In this study, we consider the simplest data model involving a single KV record. We implement a simple smart contract on the Ethereum testnet that processes the single KV record with optional on-chain storage. We use an off-chain machine running Ethereum client \texttt{geth}, to represent the SP. The two static baselines, BL1 and BL2, are implemented. We use a series of workloads with varying read-write ratios. Each workload is a repeated sequence of $X1$ writes followed by $X2$ reads (all of which are under the single data key). On the one end, we use a write-only sequence , that is, $\frac{X2}{X1}={0}$. On the other end, we use a read-intensive sequence with each write followed by $256$ reads $\frac{X2}{X1}={256}$. After driving each workload to our system, we measure the average Gas per operation on BL1 and BL2. We vary the read-to-write ratio ($\frac{X2}{X1}$) and report BL1 and BL2's Gas per operation in Figure~\ref{fig:prel:mhtdataset}.

It is clear that as the workload changes from the write-only sequence to read-intensive ones, there is a tradeoff between the two static baselines. When the workload is write-only, BL1 achieves lower Gas per operation than BL2, with cost saving more than $100\times$. When the workload becomes about every $1.5$ read per write (i.e., $\frac{X2}{X1}={1.5}$), the two approaches cost equal Gas. When the workload is more read intensive, such as $\frac{X2}{X1}={256}$, BL2's Gas per operation is $\frac{1}{7}$ of BL1's. 

While having a data replica on the blockchain is expected to shift the cost distributions between reads and writes, the striking cost difference it makes ($100\times$ and $7\times$) was surprising to us. This can be attributed to Ethereum's unique cost model: When the workload is write-only, the always-replicate baseline (BL2) incurs expensive operations to update smart-contract storage, which costs $5,000\sim{}20,000$ Gas per word; recall Table~\ref{tab:costmodel}. When the workload is read intensive, the never-replicate baseline (BL1) incurs expensive transactions to move the latest value of KV record to the blockchain, while BL2 avoids the expense by reading storage data on chain; recall Table~\ref{tab:costmodel} that a read from smart-contract storage costs $200$ Gas per word while a transaction costs a much higher $2176$ per word; let alone the initial cost of $21,000$ of an empty transaction.

\section{GRuB: System Design and Impl.}
\label{sec:grubsys}

\begin{figure*}[!htbp]
  \begin{center}
    \subfloat[{\small GRuB runs middleware across a smart-contract supported blockchain, SP and DO. The core system components are depicted by shaded boxes in the figure. In blue are data-plane components responsible for data movement and storage, running authenticated data structures (ADS) and managing replicas. In pink are control-plane components that monitor workloads, make replication decisions, and execute the decisions on the data plane.}]{%
     \includegraphics[width=0.675\textwidth]{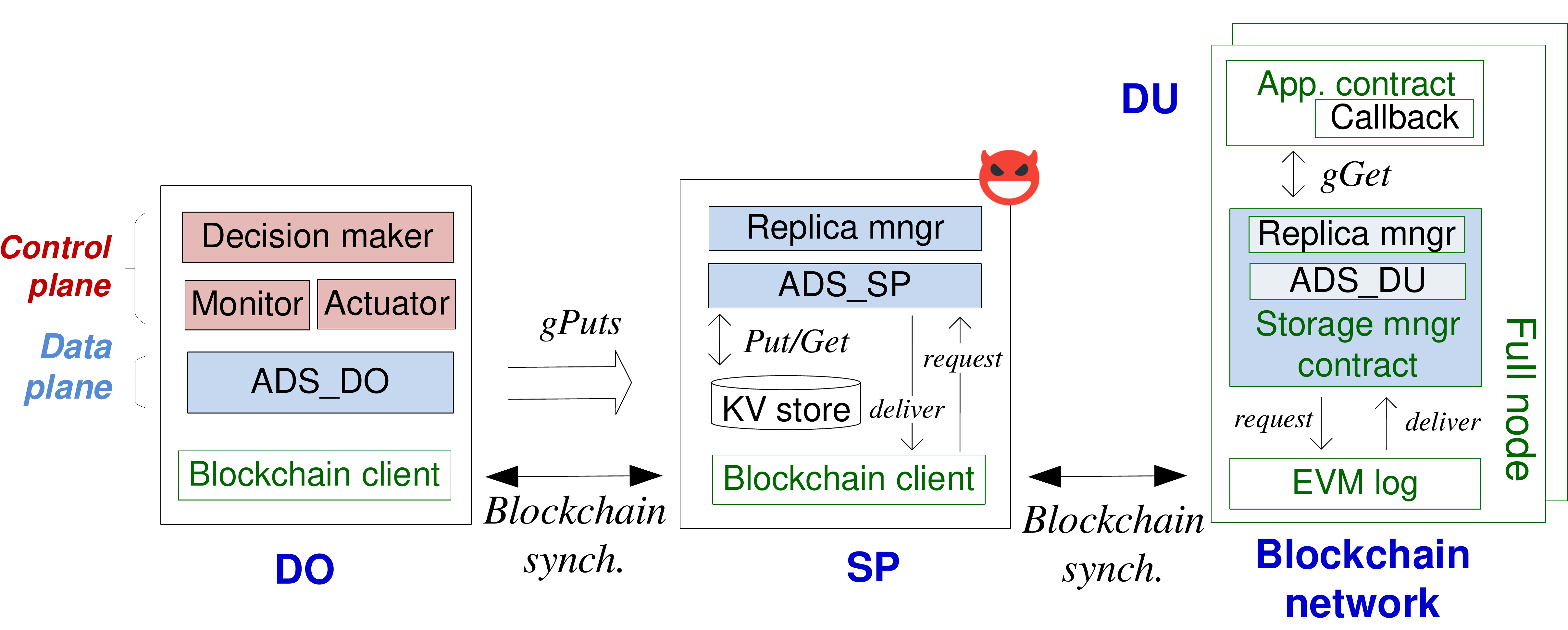}
    }
    \hspace{0.5cm}%
    \subfloat[Merkle tree in ADS\_SP]{%
      \includegraphics[width=0.19\textwidth]{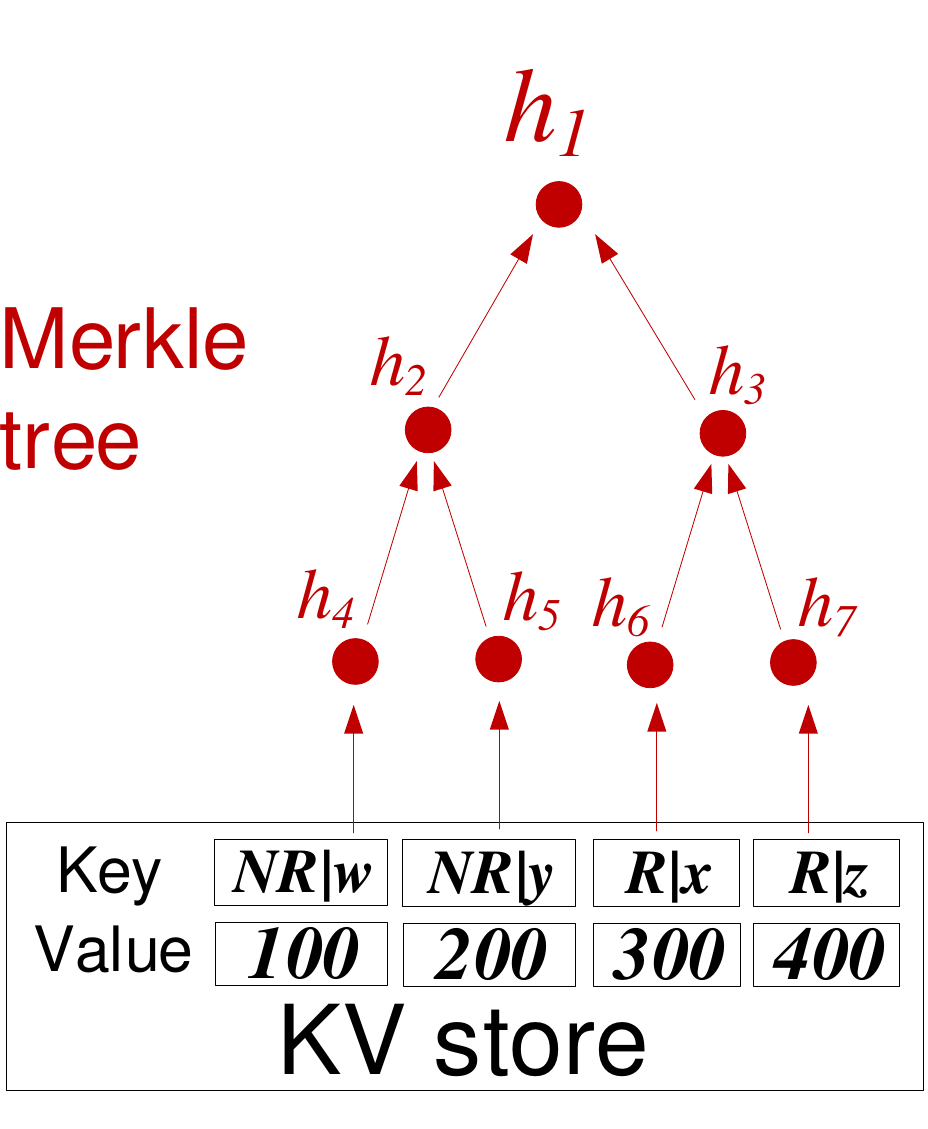}%
\label{fig:protocols:merkle0}
\label{fig:protocols:merkle}
    }%
    \end{center}
  \caption{The overview of GRuB distributed system}
\label{fig:systemoverview}
\end{figure*}

{\bf 
GRuB overview}: Recall the system model in Section~\ref{sec:systemmodel} that a trusted DO feeds data updates to the GRuB KV store, which is queried by DU smart contracts. The internal system of the GRuB consists of two ``planes'', as depicted in Figure~\ref{fig:systemoverview}: 1) A secure-data plane where the DO securely updates the KV store on GRuB by associating data updates with proofs, and a DU smart contract querying the GRuB retrieves query proofs to authenticate (non-replicated) KV records stored on the untrusted cloud provider. The data plane runs a security protocol known as authenticated data structures (ADS; which will be introduced and described in Section~\ref{sec:dataplane}) across the DO, the SP and the blockchain. 2) A control plane which monitors the workloads (data updates and reads), makes replication decisions w.r.t. individual KV records, and stores the decisions as auxiliary states in each KV record, which instructs the data plane to materialize the decisions. The control plane runs on the trusted DO and federates the traces of data reads (from the blockchain's natively logged contract-call history) and data updates. The key component of GRuB is a series of online decision-making algorithms running in the control plane. 

In this section, we describe the algorithm design of online decision maker, which is the core of GRuB's control plane (Section~\ref{sec:decisionalgorithm}), the control-plane system design (Section~\ref{sec:controlplane}), the data-plane system design (Section~\ref{sec:dataplane}), overall system properties (Section~\ref{sec:properties}) and implementation notes (Section~\ref{sec:impl}).

\subsection{Online Decision-Making: Algorithm Design \& Analysis}
\label{sec:decisionalgorithm}

In this subsection, we describe the online decision-making algorithm: Given a sequence of \texttt{gPuts} and \texttt{gGet} calls, GRuB's decision-making algorithm produces the replication decisions on affected KV records. The replication decision will be actuated as described in the next subsection. The design goal of such algorithms is to reduce the Gas cost of future data reads and writes based on the assumption that the read/write history will repeat. Intuitively, the algorithm needs to predict the future reads/writes on the KV record, estimate the cost of the two alternative decisions ($R$ or $NR$) based on the prediction, and pick the one with lower costs as the output. 
Using the existing online algorithms~\cite{DBLP:conf/icde/HuangW93} is insufficient as they are designed without awareness to GRuB's cost in Gas and the cost caused by security proofs. We propose algorithm designs and configurations that are tailored to GRuB's unique costs and that can autonomously achieve bounded worst-case Gas cost. 
In the following, we present the design and analysis of two algorithms: a ``memoryless'' online algorithm that resets its state/memory about past reads/writes upon each run, and a ``memorizing'' online algorithm that remembers the operation history across runs. 

\paragraph{Memoryless Algorithm}
\label{sec:algo:memoryless}
The memoryless algorithm for replication decision making is described in Algorithm~\ref{alg:memoryless}. 
The algorithm internally maintains a list of counters, each for a $NR$ record. The counter counts the number of consecutive reads on the data record that are received since the last write. The algorithm iterates through the read/write trace. Upon a write on a record, say $\langle{}k,v\rangle{}$, the algorithm resets the counter of record $\langle{}k,v\rangle$ back to zero and updates the record's $NR$. Upon a read on a $NR$ record, it increments its counter. When the counter reaches a preset parameter, $K$, the algorithm changes the record's state from $NR$ to $R$ and removes the data record from the list of counters.

\begin{algorithm}[h]
\caption{MemorylessRepl($ops$, $count$, $states$)}\small 
\label{alg:memoryless}
\begin{flushleft}
\hspace*{\algorithmicindent} \textbf{Input}: read/write operations $ops$, read count $count$, and the replication states $states$
 \\
\hspace*{\algorithmicindent} \textbf{Output}: updated replication states $states$ 
\end{flushleft}

\begin{algorithmic}[1]
\ForAll{$o \in ops$} 
  \If{$o.isWrite()$}
      \State $count[o.key] = 0$; $states[o.key].set(NR)$;
  \Else
      \If{$count[o.key] < K$} \State $count[o.key]++$; 
      \EndIf
      \If{count[o.key] $\geq{}K$} \State $states[o.key].set(R)$;
      \Else \State $states[o.key].set(NR)$;
      \EndIf
  \EndIf
\EndFor
\end{algorithmic}
\end{algorithm}

{\bf Algorithm analysis}: The memoryless algorithm in Algorithm~\ref{alg:memoryless} has competitiveness (w.r.t. the worst-case Gas) bounded by $1+K\frac{C_{read\_off}}{C_{update}}$. Here, $C_{update}$ is the Gas to update a byte on the blockchain storage, and $C_{read\_off}$ is the unit Gas to send one byte data from off-chain to the blockchain. The algorithmic competitiveness is analyzed in Appendix~\ref{sec:algo:analysis}.

{\bf Parameter configuration}: Parameter $K$ decides the performance of memoryless algorithm. To bound the worst-case Gas, we can set $K$ to be the following to make the algorithm $2$-competitive:

\begin{equation}
\label{eqn:memoryless:parameter}
K=C_{update}/C_{read\_off}
\end{equation}


{\color{black}
Note that Equation~\ref{eqn:memoryless:parameter} implies a static value for $K$. In a dynamic replication scheme, using static $K$, while seemingly counterintuitive, has the benefit of bounded competitiveness and can also result in actual workload-adaptive cost behavior (as will be evaluated in Section~\ref{sec:eval} and particularly in Figure~\ref{fig:macro:mixedycsb}). There can be other policies to configure $K$, including setting $K$ dynamic and adaptive to the workload for lower Gas. A comprehensive study of $K$ configuration policies is out of scope of this work, the main goal of which is providing a {\it mechanism} evaluated by selected policies.}

\paragraph{Memorizing Algorithm}
\label{sec:algo:memorizing}

\begin{algorithm}[h]
\caption{MemorizingRepl($ops$, $rCount$, $wCount$, $states$)}\small 
\label{alg:memorizing}
\begin{flushleft}
\hspace*{\algorithmicindent} \textbf{Input}: read/write operations $ops$, read counts $rCount$, write counts $wCount$ and the replication states $states$
 \\
 \hspace*{\algorithmicindent} \textbf{Output}: updated replication states $states$ 
\end{flushleft}
\begin{algorithmic}[1]
\ForAll{$o \in ops$} 
  \If{$o.isWrite()$} $wCount[o.key]++$;
  \Else{} $rCount[o.key]++$;
  \EndIf
  \If{$wCount[o.key]*K'+D <= rCount[o.key]$} 
      \State $states[o.key].set(R)$;
  \EndIf
  \If{$wCount[o]*Y-K' > rCount[o.key]$}
      \State $states[o.key].set(NR)$;
  \EndIf
\EndFor
\end{algorithmic}
\end{algorithm}

In practice, workloads exhibit temporal locality and can have repeated sequences of read/write operations. 
The memoryless algorithm does not capture the temporal locality in the workload by forgetting the past operation history. We propose a memorizing algorithm that exploits the temporal locality in workloads by memorizing the decisions made for similar operations in the past. The memorizing algorithm takes as input the trace of reads and writes. Note that unlike the memoryless algorithm, the memorizing algorithm takes in the trace of on-chain data reads.

The algorithm, described in Algorithm~\ref{alg:memorizing}, maintains two counters for each data record, $rCount$ and $wCount$. $rCount$ ($wCount$) increments when the algorithm, iterating through the read/write trace, encounters a read (write) operation. The algorithm checks two conditions upon each read/write operation: If the condition holds, $wCount*K'+D<=rCount$, the record's state is updated from $NR$ to $R$. Here, $D$ is a time window in the past the algorithm looks into to characterize the current workload and to predict the future one. It also resets $wCount$ to zero and reduces the value of $rCount$ to $D$. If the condition holds, $wCount*K'-D>=rCount$, the record's state is updated from $R$ to $NR$. It also resets $rCount$ to zero and reduces the value of $wCount$ to $D/K'$. 

{\bf Parameter configuration}: Similar to the memoryless algorithm, parameter $K'$ is set to the ratio of on-chain write cost to off-chain read cost. $K'=C_{write}/C_{read\_off}$. The other parameter $D$ determines how sensitive the replication state is to the workload. A small $D$ leads to frequent changes of replication state, while a large $D$ leads to a stable setting of replication state.

{\bf Algorithm analysis}: As analyzed in 
Appendix~\ref{sec:algo:analysis}, 
the memorizing Algorithm~\ref{alg:memorizing} is $\frac{4D+2}{K'}$-competitive.

\subsection{System Control Plane}
\label{sec:controlplane}

The previous subsection describes the online decision-making algorithms and their analysis. This subsection describes how to {\it execute} the algorithm in the control plane of GRuB. The control plane runs on the DO and is depicted in Figure~\ref{fig:systemoverview}. It consists of three essential components: a workload monitor that collects the trace of data reads and writes, the algorithm executor that executes the online algorithm with the monitored trace, and a decision actuator that stores the decisions along with the records in the KV store. 

Concretely, the workload monitor running on the DO federates the trace of data updates that occur locally and the trace of data reads from the blockchain history. Here, we consider that the blockchain has a builtin support to log smart-contract invocations, as is the case in Ethereum. The DO runs a blockchain client in full synchronization with other blockchain nodes; the client stores the contract-invocation history, from which the call sequence of \texttt{gGet}'es can be accessed. In practice, the DO can run a light blockchain client such as Simplified Payment Verification (SPV) client without downloading the transaction history.

The algorithm output, namely replication decisions, is stored as an auxiliary state in each data record in the KV store. Given a KV record, say $\langle{}k,v\rangle{}$, its key is prefixed with an extra bit that indicates whether the record has a replica on the blockchain (i.e., state $R$) or not (i.e., state $NR$). The state bit will instruct the data-plane of the system to execute the replication decisions, accordingly (See Section~\ref{sec:dataplane}).

This design assumes a trusted blockchain client whose synchronization with a remote blockchain network is secured by external mechanisms; the client can increase the number of neighbor peers to guarantee the integrity of information synchronized (including blocks and transactions) in the case of compromised blockchain nodes. We dismissed the alternative design by receiving the trace of \texttt{gGet} from the untrusted SP which is incentivized to forge the trace and mislead the DO to make a $NR$ decision. Specifically, a $NR$ decision implies more use of the cloud service and the SP can charge higher service fee.

\subsection{System Data Plane}
\label{sec:dataplane}

This subsection describes the system data plane, in terms of write and read paths. That is, how GRuB handles batched data updates and replication-state transitions from the DO (write path) and data reads from a DU under the current replication states (read path). To guarantee the data authenticity against an adversarial SP, a security protocol, ADS, is adopted in the data plane of GRuB. We begin with a background introduction to ADS.

{\bf Preliminary on ADS}: An ADS protocol, or authenticated data structure, is a security protocol running among a data owner (ADS\_DO), an untrusted service provider (ADS\_SP) and multiple data users (ADS\_DU). In its most basic form, the ADS\_SP accepts data updates (individual KV records) from the ADS\_DO and serves exact-match queries (by data keys) issued by ADS\_DU. The security properties an ADS guarantee is the authenticity of KV records including record integrity, query completeness and freshness against an adversarial ADS\_SP who can forge, replay, omit and fork~\cite{DBLP:conf/osdi/LiKMS04} records in a query. 

An ADS protocol can be constructed in different ways~\cite{DBLP:conf/esa/Tamassia03,DBLP:journals/algorithmica/MartelNDGKS04,DBLP:conf/ccs/PapamanthouTT08,DBLP:journals/algorithmica/PapamanthouTT16,DBLP:conf/ccs/ZhangKP15,DBLP:journals/algorithmica/MartelNDGKS04,DBLP:conf/sigmod/LiHKR06} and GRuB can be easily adapted to these constructions. In our current prototype implementation, we use the common construction based on a Merkle tree. That is, the ADS\_SP constructs a Merkle tree on the dataset, each leaf storing the hash of a data record, sorted by their keys. When the ADS\_DO wants to update the dataset, she first retrieves the authentication proof of the data key to be updated from the ADS\_SP, verifies the data integrity, computes the new leaf hash, and then computes the new root hash based on the proof. The ADS\_DO can then safely send the updated data record to the ADS\_SP. For data freshness, the ADS\_DO can periodically publish the signed root hash to the SP.
When a data user, ADS\_DU, queries the dataset by a queried key, SP can serve the query by returning the matched KV record and its associated proof. The proof including the latest signed root hash from the trusted ADS\_DO can be used to verify the integrity, completeness, and freshness of the query result. 

In GRuB, the KV records are sorted by their data keys on SP. Recall that each GRuB record's data key is extended with a prefix of replication state ($R$ or $NR$). The Merkle tree on ADS\_SP is constructed on the key-sorted data layout of records. 
An example Merkle tree in GRuB is depicted in Figure~\ref{fig:protocols:merkle0} where the four KV records are first ordered by their $NR$/$R$ states and then by their actual data keys.

\noindent{\bf Write path}: 
Given a stream of data updates, DO sends a \texttt{gPuts} call every epoch. To prepare the call, DO locally batches the data updates and include them in the single \texttt{gPuts} call to be sent by the end of the epoch. Internally, the \texttt{gPuts} first notifies the control plane on DO of the latest data updates. Then, for each data update, DO runs the ADS protocol with the SP to securely update the matching KV records. 

If all KV records in this batch are in non-replicated state ($NR$) and there is no update on the replication state, the DO sends only the digest of this batch to call the \texttt{update()} function in the storage-manager smart contract. 
Note that the blockchain node on DO receiving the \texttt{update()} call would propagate it to other blockchain nodes.
If there are any KV records with replicated state ($R$), they are included in the \texttt{update()} call. If there is any state transition, either from $R$ to $NR$ or from $NR$ to $R$, such transitions are included in the \texttt{update()} call. Receiving the call, the storage-manager contract would insert a new replica to on-chain storage if there is a transition from $NR$ to $R$. It would evict an existing replica if there is a transition from $R$ to $NR$.

\noindent{\bf Read path}: 
Given a \texttt{gGet} call from a DU smart contract, all blockchain nodes would execute the storage-manager contract to handle the call. If the requested data key can be matched to a $R$ KV record replicated on the blockchain, the storage manager simply returns the record into the callback function. 
Otherwise, it emits an event recorded in the Ethereum log via calling our \texttt{request} function. The event can be captured externally by a watchdog service on SP. Specifically, the \texttt{request} event is recorded on all Ethereum nodes including the client running on SP. The SP runs an external daemon process (watchdog) that spins on the log to wait for a \texttt{request} event. The event triggers the SP to query its local KV store for the requested record before sending it back to the storage-manager contract via calling the \texttt{deliver} function. 
The \texttt{deliver} function verifies the integrity of the KV records from off-chain before invoking the callback with the verified KV record.

\definecolor{mygreen}{rgb}{0,0.6,0}
\lstset{ %
  backgroundcolor=\color{white},   
  basicstyle=\scriptsize\ttfamily,        
  breakatwhitespace=false,         
  breaklines=true,                 
  captionpos=b,                    
  commentstyle=\color{mygreen},    
  deletekeywords={...},            
  escapeinside={\%*}{*)},          
  extendedchars=true,              
  keepspaces=true,                 
  keywordstyle=\color{blue},       
  language=Java,                 
  morekeywords={*,...},            
  numbers=none,                    
  numbersep=5pt,                   
  numberstyle=\small\color{black}, 
  rulecolor=\color{black},         
  showspaces=false,                
  showstringspaces=false,          
  showtabs=false,                  
  stepnumber=1,                    
  stringstyle=\color{mymauve},     
  tabsize=2,                       
  caption={GRuB's storage-manager smart contract},                  
  label={lst:storagemngr},
  moredelim=[is][\bf]{*}{*},
}
\begin{lstlisting}
contract GRuB.StorageManager {
  bytes32 rootHash;
  mapping(uint256=>uint256) KVReplicas;
  function gGet(uint256 key, uint256 callback){
    uint256 value = KVReplicas[key];
    if(value != null) callback(key, value);
    //request() emits an EVM log event
    request(key, deliver, callback);}

  function deliver(uint256 key, uint256 value, bool replicate, uint256 proof, uint256 callback){
    if(!verify(key,value,proof,rootHash)) return false;
    if(replicate) KVReplicas[key] = value;
    callback(key,value);}

  function update(uint256[] keys, uint256[] values, uint256 digest){
    if(msg.sender = DO) rootHash = digest;
    for(int i = 0; i < keys.length; i++){
      if(values[i].replicate) KVReplicas[keys[i]]=values[i];
      else delete KVReplicas[keys[i]]; }}}
\end{lstlisting}

The pseudo code of storage-manager smart contract is described in Listing~\ref{lst:storagemngr}. 
The more detailed data-plane workflow is described in 
Appendix~\ref{appendix:sec:dataplane}.

\newcommand{\tangSide}[1]{\todo[caption={},color=cyan!20!]{{\tiny Note: #1}}}
\setlength{\marginparwidth}{1.5cm}

{\color{black}
\subsection{Protocol Consistency}
\label{sec:properties}
\label{sec:freshness} 

In this section, we present the consistency of GRuB protocol and leave more formal proofs to 
Appendix~\ref{appdx:properties}.

To describe the protocol consistency, we assume a hypothetical global clock synchronized across the DO and all blockchain nodes. Note that this clock is used as a tool for protocol analysis and is not required in the actual implementation of GRuB. 

{\bf
Blockchain \& GRuB model}: In a vanilla blockchain, it takes $Pt$ time units to propagate a transaction to all nodes in the blockchain network. It takes an average of $B$ time units to produce a block. Only after $F$ blocks are produced, a transaction is considered finalized in the blockchain network. For instance, 
in Ethereum, $F$ is $250$ and $B$ is $10\sim{}19$ seconds~\cite{wood2014ethereum}.

In GRuB, an epoch $E$ is the time interval in which the DO waits and batches data updates in a transaction. 

{\bf Consistency between gPut and gGet}:
Suppose at time $t_1$ the DO submits a \texttt{gPut(k,v)} and at time $t_2$ a blockchain node $Ni$ executes \texttt{gGet(k)}. After $t_2+P_t+B\cdot{}F$, assume the execution of \texttt{gGet(k)} is finalized on the blockchain.

Particularly, when the record \texttt{gGet(k)} accesses is not replicated ($NR$), time $t_2$ refers to when the internal call of \texttt{gGet(k)} is being entered and returned by the blockchain node (the synchronous execution finishes instantly). When the record \texttt{gGet(k)} accesses is not replicated ($NR$), \texttt{gGet(k)} is executed asynchronously and is called back by a \texttt{deliver} transaction. In this case $t_2$ refers to when the \texttt{deliver} transaction is executed on node $Ni$.

\begin{theorem}[Non-deterministic ordering of concurrent gPut/gGet]
If $t_1<t_2<t_1+E+P_t+B\cdot{}F$, \texttt{gPut(k,v)} is said to occur concurrently with \texttt{gGet(k)}. With GRuB, the ordering between concurrent \texttt{gPut(k,v)} and \texttt{gGet(k)} is non-deterministic and is the same across all blockchain nodes after $t_2+P_t+B\cdot{}F$.
\end{theorem} 

Suppose a \texttt{gGet(k)} issued by a DU smart contract at local time $t$ on blockchain node $Ni$ returns a set of KV records $qs$. Query result $qs$ is fresh, w.r.t. delay $d$, if all KV records matching key $k$ and updated on data owner DO before $t-d$ are included in $qs$. Here, it assumes a global clock synchronized across the DO and any blockchain nodes $Ni$. 
Note that query freshness also implies query completeness here.

\begin{theorem}[Epoch-bounded query freshness between sequential gPut/gGet]
If $t_2>t_1+E+P_t+B\cdot{}F$, \texttt{gPut(k,v)} is said to occur sequentially after \texttt{gGet(k)}. 
Given a \texttt{gGet} sequentially after a \texttt{gPut}, GRuB guarantees the \texttt{gGet} query freshness. Here, the parameters are epoch $E$, block time $B$, propagation delay $Pt$ and the number of blocks needed for finality $F$.
\end{theorem}

\label{sec:latency} 
{\bf Supporting delay-sensitive applications}: GRuB incurs a maximum delay of $E$ to feed data to the blockchain.
Recall that in baseline BL2, data updates are sent directly, without batching, to the blockchain. BL2 guarantees the \texttt{gGet} query freshness w.r.t. delay $Pt+F\cdot{}B$.
Applications with the urgent need to feed data can be supported by BL2 where an individual data update is fed to the blockchain immediately after the DO produces it. Note that one can retrofit BL2 to GRuB for supporting these applications where data updates are selected to opt for BL2.
}

\subsection{Implementation Notes}
\label{sec:impl}

We have built a prototype of GRuB with Ethereum and a Google LevelDB~\cite{me:leveldb} instance. Note that GRuB's design is generally applicable to any storage service exposing a KV store interface (e.g., Amazon S3), an IaaS cloud service allowing user-deployed code (e.g., Amazon EC2) and any blockchains supporting smart contracts. 
In the prototype, the storage-manager smart contract is implemented in solidity~\cite{me:sol}. The off-chain code is written in Python. In particular, the replica manager and ADS protocol relies on a Python binding to interact with the underlying LevelDB~\cite{me:leveldb}.
In practice, we use the suggested transaction fee (e.g., $21000$ Gas) and Gas price (i.e., 2 GWei) in the evaluation (Section~\ref{sec:cases} and Section~\ref{sec:eval}), which are sufficient for Ethereum Ropsten testnet~\cite{me:ropstern} to accept our transactions. How to set Gas price under more adversarial settings such as DoS attacks is out of the scope of this paper.

\section{Case Studies}
\label{sec:cases}
We have built two real applications on GRuB. One is an Ether-backed stablecoin based on a price feed by GRuB and the other is a cross-chain token exchange between Ethereum and Bitcoin based on a BtcRelay style side-chain feed. 

\subsection{Stablecoins based on Price Feeds}
\label{sec:cases:stablecoin}

Recall that indirectly-backed stablecoins require feeding the price of the asset that backs the stablecoin. For instance, in stablecoin platform MakerDAO, each currency unit, a DAI, is pegged and redeemable to one-USD worth of Ether. Issuing and redeeming DAI requires Ether price feeds. We build a GRuB-based price feed and use it to support a custom stablecoin SCoin that simulates a simpler DAI.

Specifically, we build a price feed based on GRuB where the KV records store the prices of different assets including Ether.
SCoin is implemented as a custom ERC20 token whose supply (in terms of token issuance and redemption) is controlled by a smart contract we build, listed as SCoinIssuer.
The smart contract issues SCoins upon receiving Ether payments from an external buyer (i.e., \texttt{issue} function), and upon a seller's request to redeem an SCoin, transferring one-USD worth of Ether to the seller before destroying the SCoin (i.e., \texttt{redeem} function). To make sure SCoin is pegged and redeemable to one USD, the smart contract needs to read the Ether price at the time of issuance and redemption, as well as requiring over-collateralization and locking up remaining Ether. This implements a minimalist MakerDAO based on the working example in~\cite{10.1145/3387945.3388781}. 

{\bf Cost evaluation}: For Gas evaluation, we implemented three price feeds, including GRuB and the two static baselines (BL1 and BL2). We used the call trace of a real price feed, ethPriceOracle~\cite{me:ethpriceoracle}. Recall Section~\ref{sec:motivateapps} that this trace records the Ether-price updates and reads from April 25th, 2018 to April 30th, 2018. In our experiment, we set up multiple assets in the price feed: In practice, there are many assets that can be used to back a stablecoin, such as more than 2500 tokens~\cite{me:erc20list} just on Ethereum, fiat currencies (e.g.,USD, Japanese Yen, Euro) and various commodities (gold). We thus set up a KV store of $4096$ records in the price feed, each presenting an asset and its price ($\langle{}asset\_name, price\rangle{}$). In this setup, a \texttt{gPuts} batches price updates of 10 assets, which we use duplicates of the Ether price updates. 
Each \texttt{peek()} call in the trace issues 
a \texttt{gGet} invocation with 
with a callback to SCoinIssuer's \texttt{issue()} or \texttt{refund()}, at the equal chance.
By this means, we drive the call trace into GRuBPriceFeed and SCoinIssuer. 

\begin{figure}[!ht]
  \centering
  \includegraphics[width=0.425\textwidth]{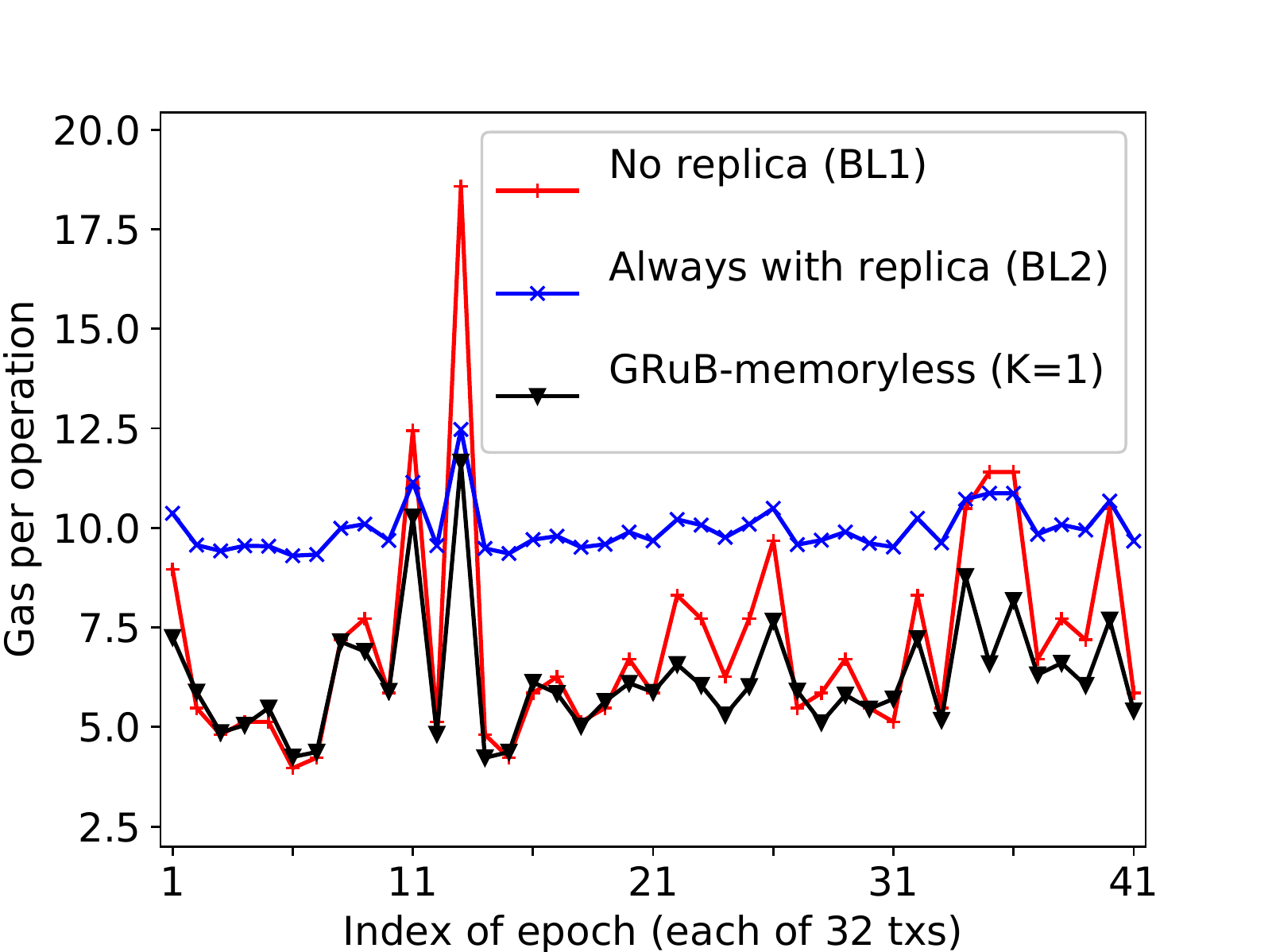}%
  \caption{Gas under the 5-day trace (ethPriceOracle)}
  \label{fig:gas:ethpricefeed}%
\end{figure}

\begin{table}[!htbp] 
\caption{Gas at the data-feed layer and Gas of the end application: $M$ denotes million of Gas.}
\label{tab:gas:ethprice}\centering{\small
\begin{tabularx}{0.35\textwidth}{ |X|c|c| }
  \hline
  & Price feed & SCoinIssuer \\ \hline
  BL1 & $83M$ ($+64\%$) & $86M$ ($+67\%$) \\ \hline
  BL2 & $55M$ ($+11\%$)  & $56M$ ($+8.7\%$) \\ \hline
  GRuB & $50.6M$ & $51.7M$ \\ \hline
\end{tabularx}
}
\end{table}

The result illustrated in Figure~\ref{fig:gas:ethpricefeed} shows that GRuB consistently achieves the lowest Gas per operation among the three. Most of the time, BL1 and GRuB achieve lower Gas than BL2. The exception is around epoch 11 when it involves more data reads that put BL1 at disadvantage. Even in this case, GRuB achieves lower Gas than BL2. 

Table~\ref{tab:gas:ethprice} shows the Gas cost at the data-feed layer and in the end application (SCoinIssuer). It can be seen while SCoinIssuer adds Gas due to application-specific logic, the Gas saving at the data feeding layer is still quite significant.

{\color{black}
\subsection{BtcRelay Side-chains and Pegged Tokens}
\label{appendix:case:sidechains}

BtcRelay feeds Bitcoin blocks to Ethereum and is an important building block for Bitcoin-pegged tokens on Ethereum. We use GRuB to enable BtcRelay by storing the mappings of block hash and Bitcoin block header in the KV store. The DO runs a trusted off-chain Bitcoin client that gets notified every time a Bitcoin block is found. 

Based on this data feed, we build a Bitcoin-pegged ERC20 token as an application. The DU smart contract is a simple ERC20 token that supports the operations of mint and burn that consume Bitcoin blocks from the feed: A token-mint (token-burn) operation requires verifying the inclusion of a Bitcoin-deposit (Bitcoin-redeem) transaction against recent Bitcoin blocks from the feed. 

{\bf Building benchmark: Methodology}: We collected the trace of transactions to mint/burn eight Bitcoin-pegged tokens known from etherscan.io~\cite{me:bitcoinpegged}. The transactions are obtained from Ethereum ETL service on Google BigQuery~\cite{me:ethtrace:bigquery} and create a token-contract workload benchmark using the method detailed below.

\begin{figure}[!ht]
  \begin{center}
     \includegraphics[width=0.35\textwidth]{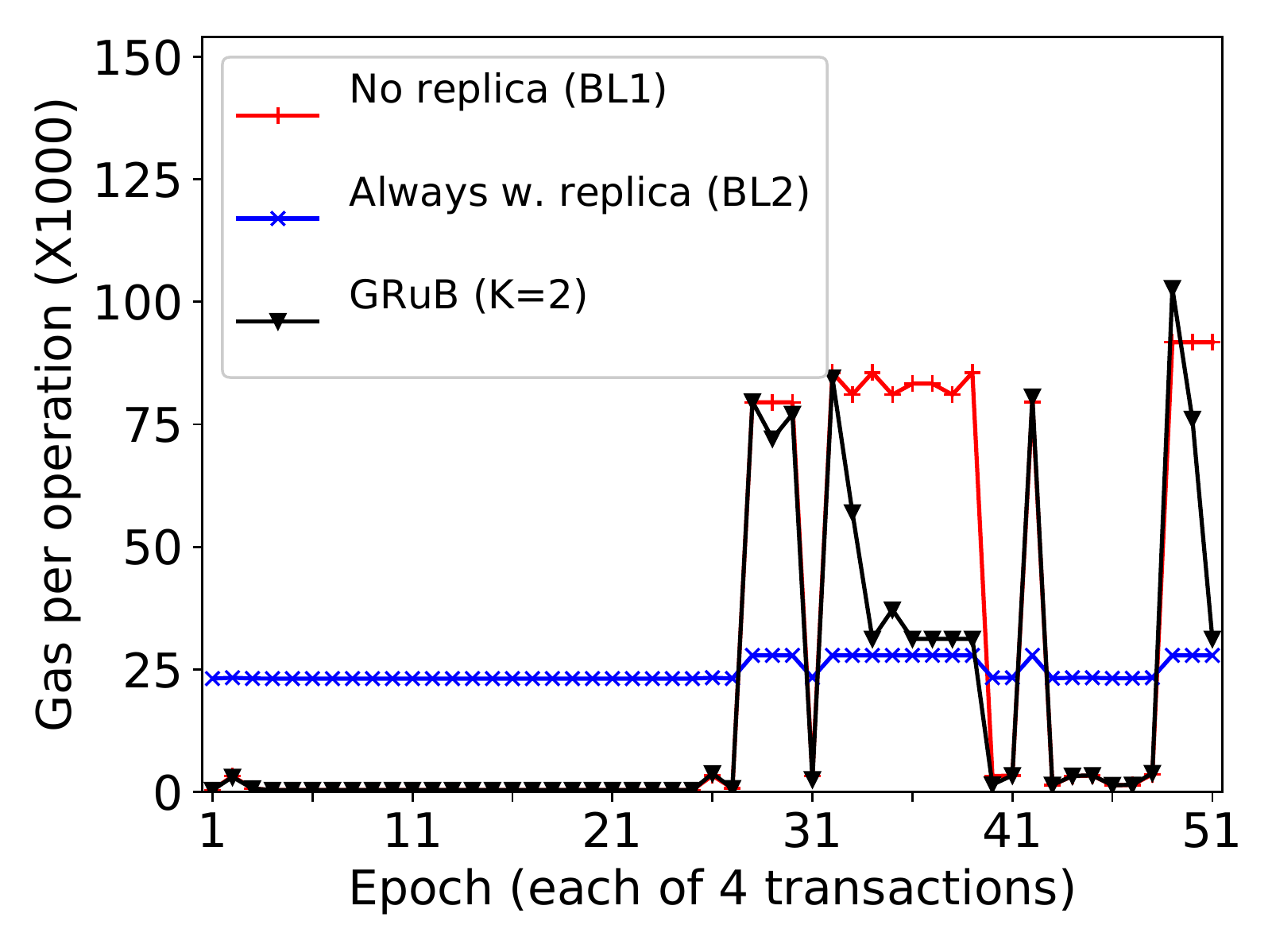}
  \end{center}
  \caption{GRuB under the BtcRelay trace}
  \label{fig:grub:btcrelay}%
\end{figure}

{\bf Experiment results}: 
We measure the Gas cost by GRuB under the workload. We set up an epoch that contains four transactions and drive the established benchmark to our experiments. Particularly, unlike the ethPriceOracle, the BtcRelay workload does not overwrite existing records, but instead appends updates to them. We configure GRuB with reusable storage upon replicating a record. To make the room, previously replicated records are evicted if they are not accessed for a long time.

The result of Gas cost per operation is reported in Figure~\ref{fig:grub:btcrelay}. 
The trace  of the first 25 epochs is write-intensive. In this phase, BL1 outperforms BL2, and GRuB converges to BL1. From epoch 26 to epoch 50, the trace becomes more read-intensive. And BL2 outperforms BL1, and GRuB gradually converges to BL2 (at epoch 34). Overall, GRuB's Gas saving is $56.7\%/14.5\%$ compared with BL1/BL2.
}

\section{Cost Evaluation}
\label{sec:eval}

This section presents the experiments for evaluating the Gas of GRuB. Specifically, our experiments are designed to answer the following questions:

1. Whether and how fast will GRuB converge to changing workloads?

2. How sensitive is GRuB's cost to the various parameters that GRuB exposes? 

We perform experiments under microbenchmarks (Section~\ref{sec:convergence}) and macro-benchmarks with YCSB~\cite{DBLP:conf/cloud/CooperSTRS10} (Section~\ref{sec:macro}).


\subsection{Microbenchmarks: Converged Gas under Repeating Workloads}
\label{sec:convergence}

\begin{figure}[!ht]
\begin{center}
\includegraphics[width=0.425\textwidth]{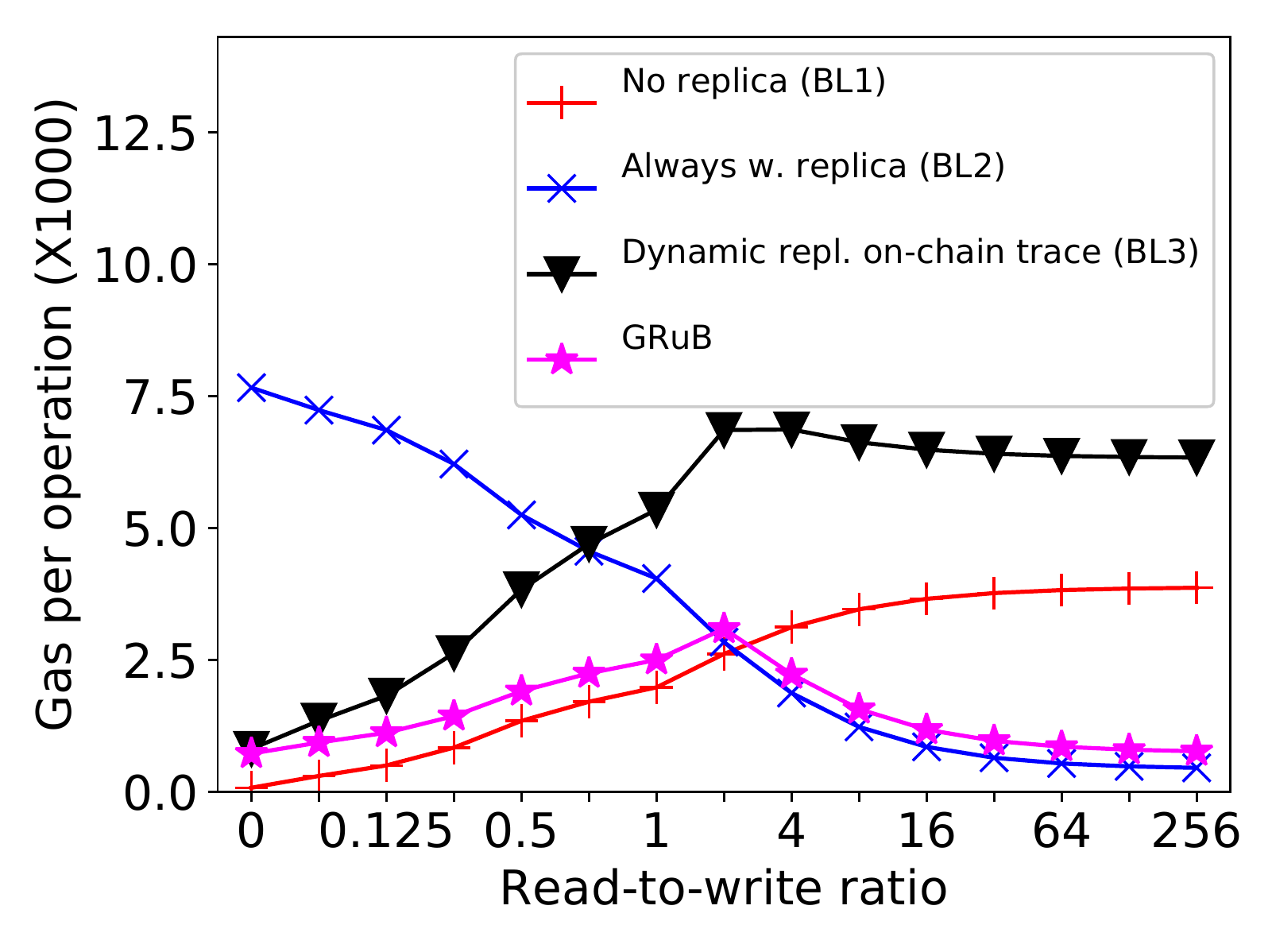}
\end{center}
\caption{GRuB's Gas with varying read-write ratios}
\label{fig:mi:rw}
\vspace{-0.10in}
\end{figure}

In this subsection, we evaluate GRuB's Gas under repeating workloads. We generate the workload that consists of repeated reads and writes under a fixed ratio. Under such workloads, GRuB makes the same decisions, and the Gas becomes converged. Our goal is to evaluate the converged Gas under different factors.  

\label{sec:converged:rw}
\noindent{\bf Read-to-write ratio}:
In this experiment, we evaluate the Gas with different read-to-write ratios. For comparison to GRuB, we consider both baselines of static data replication (i.e., BL1 and BL2). Also, we consider the two baseline designs for dynamic data replication that respectively store on the Blockchain, the trace of reads and writes, and the trace of reads. In each experiment, we drive the synthetic workload of a specific read-to-write ratio to the system and measure the total Gas. We report the average Gas per operation.

In the results reported in Figure~\ref{fig:mi:rw}, baseline BL1 (BL2) has its Gas increased (decreased) as the workload shifts from write-intensive to read-intensive. There is a crossover between BL1 and BL2 when the workload's read-to-write ratio is around 2. GRuB's Gas is slightly higher than BL1 for the read-to-write ratio smaller than 2 and is slightly higher than BL2 for the ratio larger than 2. Note that choosing the one between BL1 and BL2 with lower Gas constitutes an ideal, Gas-optimal dynamic-replication scheme. In this sense, GRuB's (converged) Gas is close to the optimal case. Comparing with the two dynamic-replication baselines, GRuB saves Gas significantly: Especially in read-intensive workloads, GRuB's Gas savings can reach an order of magnitude.
 
\begin{figure}[!ht]
  \begin{center}
    \subfloat[Memoryless and memorizing algorithms under repeating workloads]{%
        \includegraphics[width=0.245\textwidth]{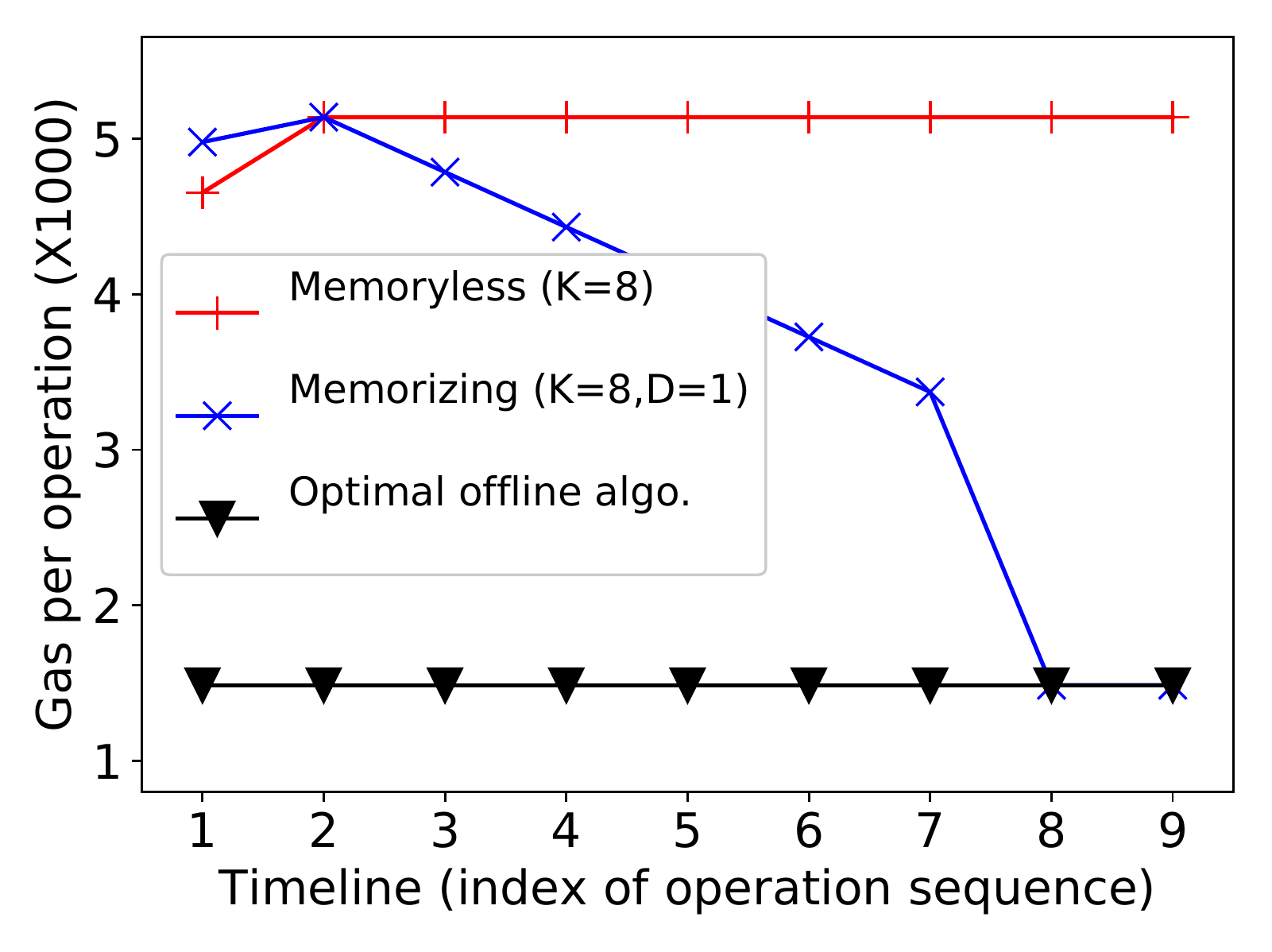}
        \label{fig:mi:vs}
    }
    \subfloat[Varying record size]{
        \includegraphics[width=0.245\textwidth]{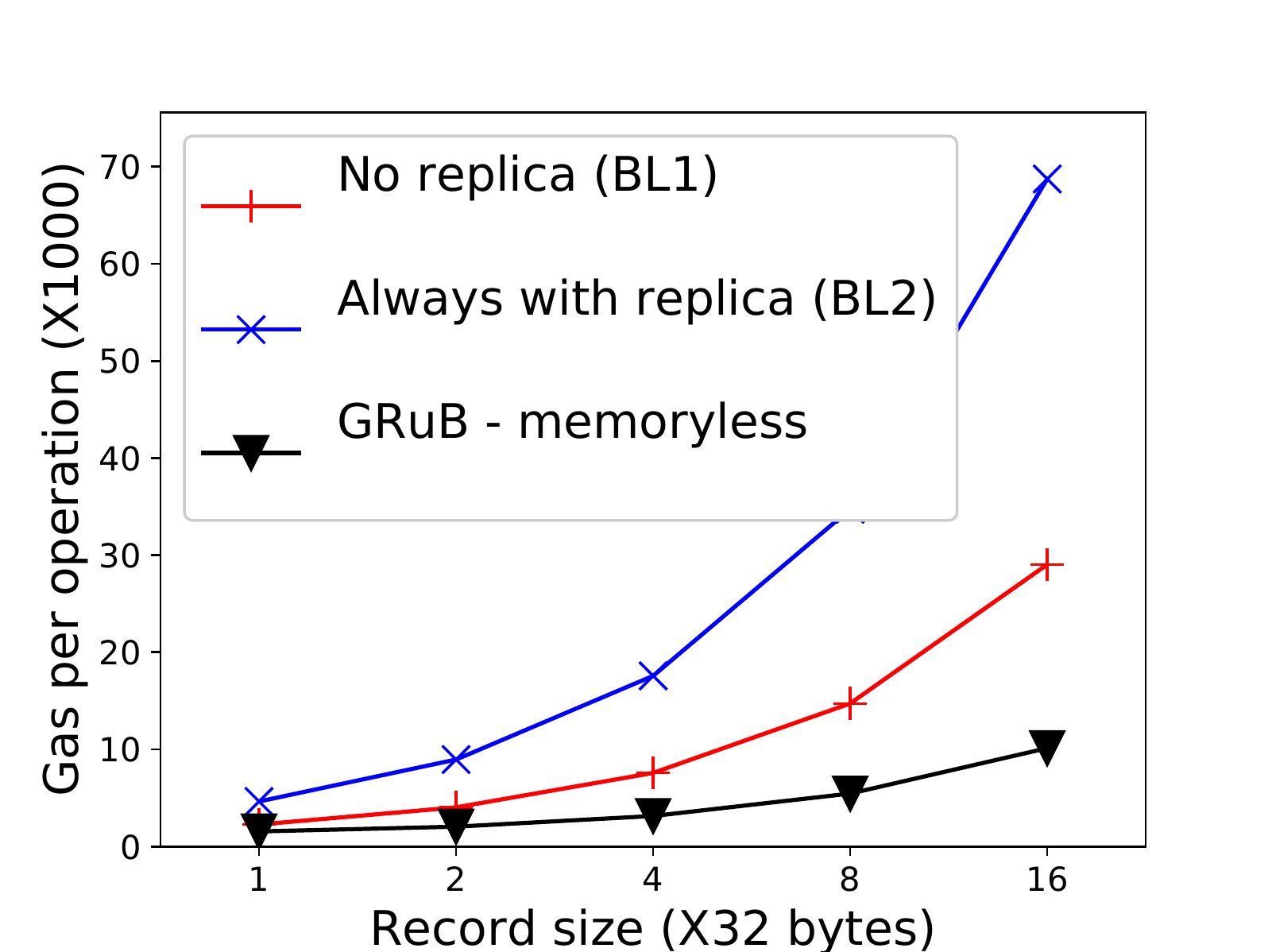}
        \label{fig:mi:recordsize}
    }%
 \end{center}
  \caption{Gas with different algorithms and record sizes}
\end{figure}

\noindent{\bf Choice of the algorithm}:
In this experiment, we evaluate how the choice of algorithms affect GRuB's Gas. Recall that we proposed two decision-making algorithms, and they differ in that the memoryless (memorizing) algorithm decides without (with) remembering the historical operations. To contrast the two algorithms to the maximal degree, we use the following experimental setting: We set parameter $K=K'$ and use the workload of read-to-write ratio $K+1$. We drive the workload to GRuB with the two different algorithms. Figure~\ref{fig:mi:vs} reports the Gas per operation along with the timeline (indexed by transactions, each encoding 32 operations). It can be seen GRuB with the memoryless algorithm incurs constant Gas, which is about 5 times higher than the optimal offline decision-making (whose Gas is calculated in a similar way with the previous experiment in Section~\ref{sec:converged:rw}). GRuB, with a memorizing algorithm, configured with $K'=8, D=1$, initially has a similar level of Gas consumption with memoryless GRuB, and then gradually reduces the Gas close to the optimal algorithm.

\noindent{\bf Record size}:In this experiment, we evaluate how GRuB's Gas is affected by data record size. We vary the record size from one word (32 bytes) to 16 words. The experiment results reported in Figure~\ref{fig:mi:recordsize} show that Gas per operation increases linearly with the record size. GRuB is cheaper in Gas than both BL1 and BL2. When the record is of 16 words, the Gas savings by GRuB reach the max, that is, $7\times$ and $3\times$ compared to BL2 and BL1, respectively.

\subsection{Macro-benchmarks on YCSB}
\label{sec:macro}
\begin{figure}[!ht]
  \begin{center}
\includegraphics[width=0.425\textwidth]{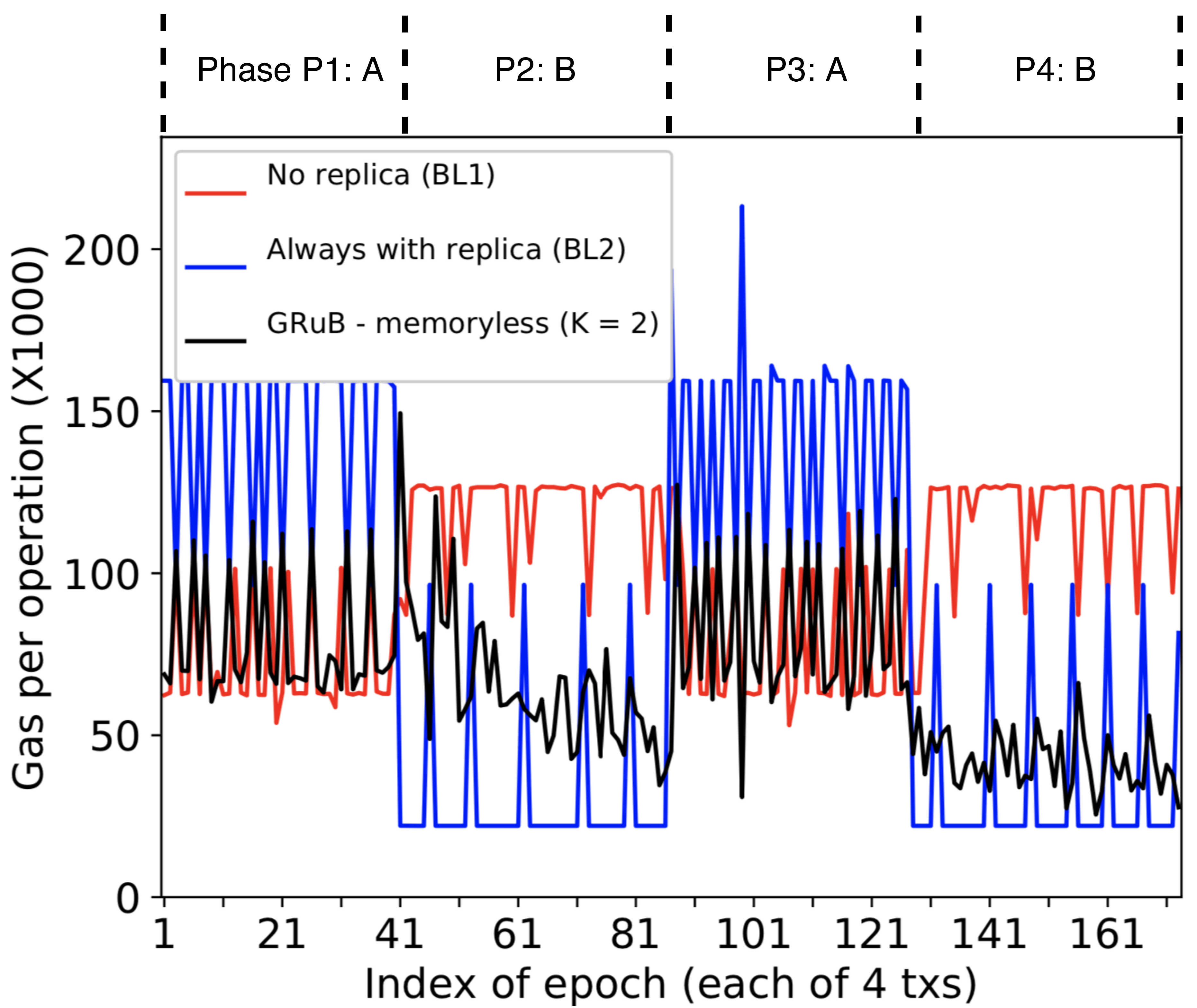}
\end{center}
\caption{GRuB under mixed YCSB workloads (A and B)}
\label{fig:ab}
\label{fig:macro:mixedycsb}
\end{figure}

\begin{table}[!htbp] 
\caption{Aggregated Gas for mixed YCSB workloads}
\label{tab:macro:mixedycsb}\centering{\scriptsize
\begin{tabularx}{0.45\textwidth}{|X|c|c|c|}
\hline
Workload & BL1 & BL2 & GRuB \\ \hline
{A, B} & 1438,130,508 (+31.6\%) & 1588,684,289 (+45.4\%) & 1092,576,982 \\ \hline
{A, E} & 1400,290,302 (+25.7\%) & 1936,114,585 (+73.8\%) & 1114,217,927 \\ \hline
{A, F} & 1746,854,231 (+54.1\%) & 1252,009,322 (+10.4\%) & 1133,858,720 \\ \hline
\end{tabularx}
}
\end{table}

This set of experiments are designed to evaluate the Gas of GRuB under mixed YCSB workloads. YCSB~\cite{DBLP:conf/cloud/CooperSTRS10} is an industrial-strength benchmark providing six KV-store workloads, codenamed from A to F, that model some real workloads in Yahoo cloud services. We use YCSB workloads to evaluate GRuB, because of the following reason: GRuB exposes the same KV-store API with most cloud-storage services and brings trustworthiness to these services. Thus, GRuB can and should be a secure alternative to host cloud workloads, especially for security-sensitive applications. In preparing GRuB macro-benchmarks, we mix multiple YCSB workloads, for instance, Workload A and E. 

In our experiments, we used three combinations: mixing Workload A and B, mixing Workload A and E, and mixing Workload A and F. Workload A/B/E/F respectively feature $50\%$ reads/$95\%$ reads/$95\%$ scans/$75\%$ reads as well as different key-distribution strategies. In each mixed workload, we pre-load $2^{16}$ KV records to the GRuB. In Workloads A,B and Workloads A,E, each KV record is set to be $1024$-bytes long. In Workload A,F, each KV record is $32$-byte long. Each experiment consists of four phases, in each of which one workload generator (i.e., A/B/E/F) is chosen to produce $4096$ operations. We report the average Gas per operation for every four transactions (or an epoch).

We report the experiments results by time-series data in Figure~\ref{fig:macro:mixedycsb} and by aggregate results in Table~\ref{tab:macro:mixedycsb}. It can be seen that in the first phase P1 (Workload A of 50\% reads), the non-replication baseline, BL1, performs better than replication baseline BL2. In Phase P1, GRuB's Gas is close to that of BL1. In the second phase, when the workload switches to Workload B (of 95\% reads), the replication baseline BL2 achieves lower Gas than BL1. In Phase P2, GRuB's Gas is lower than BL1 but higher than BL2: Especially at the beginning of P2, GRuB incurs high Gas because of the decision to replicate a KV record being read. Phase P3 is similar to P1. In Phase P4, GRuB's Gas is much lower than P2 because records being read in this phase may already be replicated in Phase P2. The aggregated result, the Gas per operation averaged overall operations, is reported in Table~\ref{tab:macro:mixedycsb}, where GRuB saves $32\%$ Gas of BL2 and $46\%$ Gas of BL1. {\color{black} Note that Figure~\ref{fig:macro:mixedycsb} shows a lower rate of saving than Figure~\ref{fig:mi:rw}, because the YCSB workloads tested here are with more restrictive read-write rates than the synthetic workloads used in Figure~\ref{fig:mi:rw}.}

\section{Related Work}
In this section, we describe two research bodies most relevant to our work: cost-effective blockchain applications and workload-aware data replication schemes.

{\bf Cost-effective blockchain applications}: It is well known that blockchain has limited throughput in handling transactions~\cite{DBLP:conf/fc/CromanDEGJKMSSS16} and incurs high unit cost to execute smart contracts. To reduce the blockchain costs, general approaches are developed by focusing on a permissioned setting~\cite{DBLP:journals/pvldb/GuptaRHS20,DBLP:conf/icdcn/BuchnikF20,DBLP:conf/srds/RuschBK19,DBLP:conf/dsn/SousaBV18}, or by sharding the blockchain and other layer-one designs~\cite{DBLP:conf/sp/Kokoris-KogiasJ18,DBLP:conf/ccs/LuuNZBGS16,DBLP:conf/srds/RuschBK19}. Unfortunately, these new blockchain designs cannot be integrated with an operational blockchain and are known to be difficult to deploy at scale.

A more practical design paradigm, also more relevant to our work, is the layer-two approaches that aim at reducing the use of blockchain in a domain application, without changing the underlying blockchain mechanisms. Notably, payment channels and networks~\cite{me:lightning,DBLP:journals/corr/MillerBKM17,DBLP:journals/corr/abs-1804-05141} process multiple micro-payments off-chain with issuing two Bitcoin transactions. There are payment networks adopted in practice, such as Lightning Networks in the Bitcoin mainnet~\cite{me:lightning}. 
Teechain~\cite{DBLP:conf/sosp/LindNEKSP19} is a payment network that offloads the detection of participant misbehavior to trusted hardware off-chain, further reducing the involvement of blockchain and improving the application throughput. 
Similarly, Tesseract~\cite{DBLP:conf/ccs/BentovJ0BDJ19} employs off-chain trusted hardware to facilitate the payments and exchanges with lower-level involvement of blockchains, to achieve real-timeness and higher throughput. 
Beyond the simple application of payments, there are authenticated data structures proposed, such as TPAD~\cite{DBLP:conf/fc/TangX0CX18} and $GEM^2$ trees~\cite{DBLP:conf/icde/Zhang0XTC19}, to enable the secure handling of database queries off-chain.
These layer-two protocols and systems have their off-chain component statically fixed and are not aware of the changing workloads. By contrast, GRuB is the first work that dynamically replicates data onto blockchain.

Testimonium~\cite{DBLP:journals/corr/abs-2004-10488} is a Gas-effective blockchain relay (or in our terminology, side-chain feed) which achieves low Gas by lazily validating blocks from a remote blockchain. Our GRuB differs from Testimonium in two senses: First, our focus on data feeding makes GRuB more generally applicable. GRuB supports applications that rely on real-world data feed (e.g., price-feed based stablecoins) that Testimonium cannot support. Second, Testimonium can be thought of as a static data-replication scheme, as it optimistically stores blocks from the remote blockchain without validation. 

TownCrier~\cite{DBLP:conf/ccs/ZhangCCJS16} is a provable-secure data feed service built on off-chain trusted hardware that connects TLS-certified websites to blockchains. The data-feed storage in TownCrier is always off-chain and it does not address the dynamic data replication as in GRuB.

Gasper~\cite{DBLP:conf/wcre/ChenLLZ17} is a compiler-based optimization pass that detects Gas-inefficiency anti-pattern in the generated contract bytecode. This language-based optimization takes into account the information at syntactic level and is not aware of any application semantics. Compared to Gasper, GRuB is modeled based on the data-feed design pattern of blockchain applications and is aware of application semantics. 

{\bf Workload-aware data replication}: In distributed databases, adaptive data replication~\cite{DBLP:journals/tods/WolfsonJH97,DBLP:conf/icde/HuangW93,DBLP:conf/sigmod/HuangSW94} has been studied: A framework has been proposed by dynamically monitoring the workload and making replication decisions based on the current workload. 
Many web applications exhibit skewed data-access patterns. MET~\cite{DBLP:conf/eurosys/CruzMMOPPV13} is a KV store management system that adapts the system configuration and cloud-resource provisioning to the current workload.
In designing P2P-based DNS services, Beehive~\cite{DBLP:conf/nsdi/RamasubramanianS04} is a proactive data replication scheme that is tailored for Zipf query distribution and achieves the constant look-up cost.
GRuB's dynamic replication scheme is inspired by these classic techniques and addresses the technical challenges when combining these classic techniques with blockchains' cost model.

\section{Conclusion}
This work presents GRuB, a dynamic data replication scheme that achieves low Gas under changing data-access workloads. GRuB runs a Gas-aware, security-centric control framework off the Blockchain. 
Evaluation shows GRuB saves Gas by upto 70\% compared with existing approaches.

\appendix

{\small
\section{Appendix: Algorithm Analysis}
\label{sec:algo:analysis}

This section analyzes the competitiveness of online algorithms, that is, the worst-case complexity compared with that of an optimal off-line algorithm. 

\begin{theorem}
\label{thm:memoryless}
Memoryless Algorithm~\ref{alg:memoryless} with parameters configured by Equation~\ref{eqn:memoryless:parameter} is $2$-competitive w.r.t. the Gas cost.
\end{theorem}

\begin{proof}
We first set up the stage by considering an ideal offline algorithm with optimal cost. This offline algorithm can know the entire sequence of reads and writes in advance, and learn the cost-optimal decision. For instance, it can check given a write, if there are more than $K$ consecutive reads that occur after it (before the next write). If so, it can replicate the record at the time of the write, instead of waiting until $K$ reads as in the online algorithm. 

For our online algorithm, the worst-case sequence of reads and writes is that every write is followed by exactly $K$ reads. This is the worst-case for our online algorithm because every data replica made by the algorithm is never read, in other words, the cost of replication is totally wasted without saving any cost (of follow-up reads). In this worst case, the cost of our algorithm is $K*C_{read\_off} + C_{update}$. 
In this case, the cost of the ideal offline algorithm is $C_{update}$. Thus, the competitiveness of our online algorithm is $1+K*\frac{C_{read\_off}}{C_{update}}$.

Plugging Equation~\ref{eqn:memoryless:parameter} in, we have a competitiveness is bounded by 
$1+\frac{C_{update}}{C_{read\_off}}*\frac{C_{read\_off}}{C_{update}}$, which is equal to $2$.
\end{proof}

\begin{theorem}
\label{thm:memorizing}
Memorizing Algorithm~\ref{alg:memorizing} is $\frac{4D+2}{K'}$-competitive.
\end{theorem}
\begin{proof}
We use the same offline algorithm as in proving Theorem~\ref{thm:memoryless}. 
We consider the following sequence of reads/writes for analyzing the worst-case of our memorizing algorithm. The read-write sequence consists of a series of sub-sequences, where the $i$-th subsequence is of $A_i$ reads and $B_i$ writes. We will set $A_i$ and $B_i$ such that the algorithm will make ``wrong'' decisions about data replication: It will decide to replicate the data record when it sees $A_i$ reads, and then not to replicate after seeing the next $B_i$ writes. Because each replication decision is followed by writes, the replica is not being read. In other words, the cost of replication is paid without any cost benefit in serving reads by replica. Each no-replication decision is followed by reads, so the follow-up reads are served at the high cost without data replica. In summary, every decision made by the algorithm does not save the cost of serving the following operations, but still incurs the cost of data replication. In this sense, this sequence serves the worst case of our algorithm.

In the $i$-th sequence, when the algorithm sees $A_i$ reads, it satisfies the in-equation $(B_1+B_2+...B_i-1) * K' <= (A_1+A_2+...+A_i) - D$; When it sees $B_i$ writes, it satisfies the in-equation $(B_1+B_2+...+B_i) * K' > (A_1+A_2+...+A_i) + D$;
Combining the two in-equations, we conclude that $A_i > 2D$, $B_i >= A_i / K'$. Finally, the general formula for the $i$-th sequence is:
$A_i = D$ when $4i = 1$; $A_i = 2D+1$ when $i > 1$; $B_i = (2D+1) / K'$.

The cost of the $i$-th sequence in our algorithm is $A_i *C_{read\_off} + C_{update} + (B_i - 1) * C_{update}$, and the cost of the ideal offline algorithm is $C_{update}$; thus the competitiveness of the memorizing algorithm is $A_i * C_{read\_off}/C_{update} + B_i$, since $C_{read\_off}/C_{update}$ equals $1 / K'$, the competitiveness is $(4D+2) / K'$.
\end{proof}

\label{appendix:case:pricefeed}
}

\clearpage

\bibliographystyle{abbrv}
\bibliographystyle{ACM-Reference-Format}
{\scriptsize
\bibliography{bkc,bkc2,lsm,misc,crypto,ads,sgx,txtbk,distrkvs,yuzhetang}

\begin{thebibliography}{10}

\bibitem{me:amazons3}
Amazon s3: Object storage built to store and retrieve any amount of data from
  anywhere.
\newblock \url{ https://aws.amazon.com/s3/ }.

\bibitem{me:btcrelay}
A bridge between the bitcoin blockchain \& ethereum smart contracts.
\newblock \url{http://btcrelay.org/}.

\bibitem{me:ethlottery:etherscan}
Contract address of ethereumlottery.io on etherscan.
\newblock
  \url{https://etherscan.io/address/0x302fE87B56330BE266599FAB2A54747299B5aC5B
  }.

\bibitem{me:tbtc:website}
Deposit and redeem btc in defi without intermediaries.
\newblock \url{https://tbtc.network/}.

\bibitem{me:eosio}
Eos: Blockchain software architecture.
\newblock \url{ https://eos.io/ }.

\bibitem{me:erc20list}
Erc20 token list.
\newblock \url{https://bloxy.info/list\_tokens/ERC20}.

\bibitem{me:ethtrace:bigquery}
Ethereum blockchain public dataset (hosted by google bigquery).
\newblock
  \url{https://console.cloud.google.com/bigquery?project=bigquery-public-data&page=dataset&d=ethereum_blockchain&p=bigquery-public-data&redirect_from_classic=true
  }.

\bibitem{me:btcrelay:github}
Ethereum contract for bitcoin spv.
\newblock \url{https://github.com/ethereum/btcrelay}.

\bibitem{me:ethlottery:reddit}
Ethereum lottery - uses bitcoin blocks to pick the winner (via btcrelay).
\newblock
  \url{https://www.reddit.com/r/ethereum/comments/4qql6d/ethereum_lottery_uses_bitcoin_blocks_to_pick_the/
  }.

\bibitem{me:eth}
Ethereum project.
\newblock \url{https://www.ethereum.org/}.

\bibitem{me:libra}
Facebook libra.
\newblock \url{ https://libra.org/en-US/ }.

\bibitem{me:pricefeed:makerdao}
Feeds price feed oracles: External reference prices for the maker platform.
\newblock \url{https://developer.makerdao.com/feeds/}.

\bibitem{me:leveldb}
{Google LevelDB}, http://code.google.com/p/leveldb/.

\bibitem{me:keybase}
Keybase.
\newblock \url{https://keybase.io/}.

\bibitem{me:keybase:bkc}
Keybase is now writing to the stellar blockchain.
\newblock
  \url{https://keybase.io/docs/server_security/merkle_root_in_stellar_blockchain}.

\bibitem{me:lightning}
Ligntning network, scalable, instant bitcoin/blockchain transactions.

\bibitem{me:ibmmaersk}
{Maersk and IBM Introduce TradeLens Blockchain Shipping Solution},
  https://ibm.co/37oj56n.

\bibitem{me:ethpriceoracle}
Maker: Medianizer 2 on etherscan.
\newblock
  \url{https://etherscan.io/address/0x729d19f657bd0614b4985cf1d82531c67569197b#code
  }.

\bibitem{me:makerdao}
Makerdao: Digital currency that can be used by anyone, anywhere, anytime.

\bibitem{me:ropstern}
Ropstern testnet in ethereum.
\newblock
  \url{https://github.com/ethereum/go-ethereum/blob/79b11121a7e4beef0d0297894289200b9842c36c/params/bootnodes.go#L34}.

\bibitem{me:sol}
Solidity programming language.
\newblock \url{https://solidity.readthedocs.io/en/develop/}.

\bibitem{me:tbtc:etherscan}
tbtc (tbtc) token tracker on etherscan.
\newblock
  \url{https://etherscan.io/token/0x1bBE271d15Bb64dF0bc6CD28Df9Ff322F2eBD847 }.

\bibitem{me:dai:etherscan}
Token: Dai stablecoin.
\newblock
  \url{https://etherscan.io/token/0x6b175474e89094c44da98b954eedeac495271d0f?a=0x1e0447b19bb6ecfdae1e4ae1694b0c3659614e4e}.

\bibitem{me:tether:etherscan}
Token: Tether usd.
\newblock
  \url{https://etherscan.io/token/0xdac17f958d2ee523a2206206994597c13d831ec7}.

\bibitem{me:btcpegged:etherscan}
Token tracker bitcoin pegged on etherscan.
\newblock \url{ https://etherscan.io/tokens/label/bitcoin-pegged }.

\bibitem{me:bitcoinpegged}
Token tracker: Bitcoin pegged on etherscan.
\newblock \url{ https://etherscan.io/tokens/label/bitcoin-pegged }.

\bibitem{me:stablecoin}
Token tracker: Stablecoin.
\newblock \url{https://etherscan.io/tokens/label/stablecoin }.

\bibitem{me:tbtc:github}
Trustlessly tokenized bitcoin on ethereum.
\newblock \url{ https://github.com/keep-network/tbtc}.

\bibitem{DBLP:conf/cloud/Abu-LibdehPW10}
H.~Abu{-}Libdeh, L.~Princehouse, and H.~Weatherspoon.
\newblock {RACS:} a case for cloud storage diversity.
\newblock In J.~M. Hellerstein, S.~Chaudhuri, and M.~Rosenblum, editors, {\em
  Proceedings of the 1st {ACM} Symposium on Cloud Computing, SoCC 2010,
  Indianapolis, Indiana, USA, June 10-11, 2010}, pages 229--240. {ACM}, 2010.

\bibitem{DBLP:conf/usenix/AliNSF16}
M.~Ali, J.~C. Nelson, R.~Shea, and M.~J. Freedman.
\newblock Blockstack: {A} global naming and storage system secured by
  blockchains.
\newblock In A.~Gulati and H.~Weatherspoon, editors, {\em {USENIX} {ATC} 2016},
  pages 181--194. {USENIX} Association, 2016.

\bibitem{me:grub}
A.~anonymized.
\newblock Cost-effective data feeding to blockchains via workload-adaptive data
  replication (technical report 2020).
\newblock \url{https://tinyurl.com/y94rj85o}.

\bibitem{DBLP:conf/ccs/BentovJ0BDJ19}
I.~Bentov, Y.~Ji, F.~Zhang, L.~Breidenbach, P.~Daian, and A.~Juels.
\newblock Tesseract: Real-time cryptocurrency exchange using trusted hardware.
\newblock In {\em Proceedings of the 2019 {ACM} {SIGSAC} Conference on Computer
  and Communications Security, {CCS} 2019, London, UK, November 11-15, 2019},
  pages 1521--1538, 2019.

\bibitem{DBLP:journals/tos/BessaniCQAS13}
A.~N. Bessani, M.~Correia, B.~Quaresma, F.~Andr{\'{e}}, and P.~Sousa.
\newblock Depsky: Dependable and secure storage in a cloud-of-clouds.
\newblock {\em {ACM} Trans. Storage}, 9(4):12:1--12:33, 2013.

\bibitem{DBLP:conf/icdcn/BuchnikF20}
Y.~Buchnik and R.~Friedman.
\newblock A generic efficient biased optimizer for consensus protocols.
\newblock In {\em {ICDCN} 2020: 21st International Conference on Distributed
  Computing and Networking, Kolkata, India, January 4-7, 2020}, pages
  18:1--18:10, 2020.

\bibitem{DBLP:conf/wcre/ChenLLZ17}
T.~Chen, X.~Li, X.~Luo, and X.~Zhang.
\newblock Under-optimized smart contracts devour your money.
\newblock In {\em {IEEE} 24th International Conference on Software Analysis,
  Evolution and Reengineering, {SANER} 2017, Klagenfurt, Austria, February
  20-24, 2017}, pages 442--446, 2017.

\bibitem{DBLP:journals/corr/abs-1804-05141}
R.~Cheng, F.~Zhang, J.~Kos, W.~He, N.~Hynes, N.~M. Johnson, A.~Juels,
  A.~Miller, and D.~Song.
\newblock Ekiden: {A} platform for confidentiality-preserving, trustworthy, and
  performant smart contract execution.
\newblock {\em CoRR}, abs/1804.05141, 2018.

\bibitem{10.1145/3387945.3388781}
J.~Clark, D.~Demirag, and S.~Moosavi.
\newblock Demystifying stablecoins.
\newblock {\em Queue}, 18(1):39–60, Feb. 2020.

\bibitem{DBLP:conf/cloud/CooperSTRS10}
B.~F. Cooper, A.~Silberstein, E.~Tam, R.~Ramakrishnan, and R.~Sears.
\newblock Benchmarking cloud serving systems with ycsb.
\newblock In {\em SoCC}, pages 143--154, 2010.

\bibitem{DBLP:conf/fc/CromanDEGJKMSSS16}
K.~Croman, C.~Decker, I.~Eyal, A.~E. Gencer, A.~Juels, A.~E. Kosba, A.~Miller,
  P.~Saxena, E.~Shi, E.~G. Sirer, D.~Song, and R.~Wattenhofer.
\newblock On scaling decentralized blockchains - {(A} position paper).
\newblock In J.~Clark, S.~Meiklejohn, P.~Y.~A. Ryan, D.~S. Wallach, M.~Brenner,
  and K.~Rohloff, editors, {\em {FC} 2016 Workshops}, volume 9604 of {\em
  Lecture Notes in Computer Science}, pages 106--125. Springer, 2016.

\bibitem{DBLP:conf/eurosys/CruzMMOPPV13}
F.~Cruz, F.~Maia, M.~Matos, R.~Oliveira, J.~Paulo, J.~Pereira, and
  R.~Vila{\c{c}}a.
\newblock Met: workload aware elasticity for nosql.
\newblock In {\em Eighth Eurosys Conference 2013, EuroSys '13, Prague, Czech
  Republic, April 14-17, 2013}, pages 183--196, 2013.

\bibitem{DBLP:journals/pvldb/GuptaRHS20}
S.~Gupta, S.~Rahnama, J.~Hellings, and M.~Sadoghi.
\newblock Resilientdb: Global scale resilient blockchain fabric.
\newblock {\em Proc. {VLDB} Endow.}, 13(6):868--883, 2020.

\bibitem{DBLP:conf/sigmod/HuangSW94}
Y.~Huang, A.~P. Sistla, and O.~Wolfson.
\newblock Data replication for mobile computers.
\newblock In {\em Proceedings of the 1994 {ACM} {SIGMOD} International
  Conference on Management of Data, Minneapolis, Minnesota, USA, May 24-27,
  1994.}, pages 13--24, 1994.

\bibitem{DBLP:conf/icde/HuangW93}
Y.~Huang and O.~Wolfson.
\newblock A competitive dynamic data replication algorithm.
\newblock In {\em Proceedings of the Ninth International Conference on Data
  Engineering, April 19-23, 1993, Vienna, Austria}, pages 310--317, 1993.

\bibitem{DBLP:conf/sp/Kokoris-KogiasJ18}
E.~Kokoris{-}Kogias, P.~Jovanovic, L.~Gasser, N.~Gailly, E.~Syta, and B.~Ford.
\newblock Omniledger: {A} secure, scale-out, decentralized ledger via sharding.
\newblock In {\em 2018 {IEEE} Symposium on Security and Privacy, {SP} 2018,
  Proceedings, 21-23 May 2018, San Francisco, California, {USA}}, pages
  583--598, 2018.

\bibitem{DBLP:conf/sigmod/LiHKR06}
F.~Li, M.~Hadjieleftheriou, G.~Kollios, and L.~Reyzin.
\newblock Dynamic authenticated index structures for outsourced databases.
\newblock In {\em SIGMOD Conference}, pages 121--132, 2006.

\bibitem{DBLP:conf/osdi/LiKMS04}
J.~Li, M.~N. Krohn, D.~Mazi{\`e}res, and D.~Shasha.
\newblock Secure untrusted data repository (sundr).
\newblock In {\em OSDI}, pages 121--136, 2004.

\bibitem{DBLP:conf/sosp/LindNEKSP19}
J.~Lind, O.~Naor, I.~Eyal, F.~Kelbert, E.~G. Sirer, and P.~R. Pietzuch.
\newblock Teechain: a secure payment network with asynchronous blockchain
  access.
\newblock In {\em Proceedings of the 27th {ACM} Symposium on Operating Systems
  Principles, {SOSP} 2019, Huntsville, ON, Canada, October 27-30, 2019}, pages
  63--79, 2019.

\bibitem{DBLP:journals/corr/abs-2005-04377}
B.~Liu and P.~Szalachowski.
\newblock A first look into defi oracles.
\newblock {\em CoRR}, abs/2005.04377, 2020.

\bibitem{DBLP:conf/ccs/LuuNZBGS16}
L.~Luu, V.~Narayanan, C.~Zheng, K.~Baweja, S.~Gilbert, and P.~Saxena.
\newblock {A Secure Sharding Protocol For Open Blockchains}.
\newblock In {\em {CCS} 2016}, pages 17--30, 2016.

\bibitem{DBLP:journals/algorithmica/MartelNDGKS04}
C.~U. Martel, G.~Nuckolls, P.~T. Devanbu, M.~Gertz, A.~Kwong, and S.~G.
  Stubblebine.
\newblock A general model for authenticated data structures.
\newblock {\em Algorithmica}, 39(1):21--41, 2004.

\bibitem{DBLP:conf/healthcom/Mettler16}
M.~Mettler.
\newblock Blockchain technology in healthcare: The revolution starts here.
\newblock In {\em 18th {IEEE} International Conference on e-Health Networking,
  Applications and Services, Healthcom 2016, Munich, Germany, September 14-16,
  2016}, pages 1--3. {IEEE}, 2016.

\bibitem{DBLP:journals/corr/MillerBKM17}
A.~Miller, I.~Bentov, R.~Kumaresan, and P.~McCorry.
\newblock Sprites: Payment channels that go faster than lightning.
\newblock {\em CoRR}, abs/1702.05812, 2017.

\bibitem{moin2019classification}
A.~Moin, E.~G. Sirer, and K.~Sekniqi.
\newblock A classification framework for stablecoin designs, 2019.

\bibitem{DBLP:conf/ccs/PapamanthouTT08}
C.~Papamanthou, R.~Tamassia, and N.~Triandopoulos.
\newblock Authenticated hash tables.
\newblock In {\em Proceedings of the 2008 {ACM} Conference on Computer and
  Communications Security, {CCS} 2008, Alexandria, Virginia, USA, October
  27-31, 2008}, pages 437--448, 2008.

\bibitem{DBLP:journals/algorithmica/PapamanthouTT16}
C.~Papamanthou, R.~Tamassia, and N.~Triandopoulos.
\newblock Authenticated hash tables based on cryptographic accumulators.
\newblock {\em Algorithmica}, 74(2):664--712, 2016.

\bibitem{poon2016bitcoin}
J.~Poon and T.~Dryja.
\newblock The bitcoin lightning network: Scalable off-chain instant payments.
\newblock 2016.

\bibitem{DBLP:conf/nsdi/RamasubramanianS04}
V.~Ramasubramanian and E.~G. Sirer.
\newblock Beehive: {O(1)} lookup performance for power-law query distributions
  in peer-to-peer overlays.
\newblock In {\em 1st Symposium on Networked Systems Design and Implementation
  {(NSDI} 2004), March 29-31, 2004, San Francisco, California, USA,
  Proceedings}, pages 99--112, 2004.

\bibitem{DBLP:conf/srds/RuschBK19}
S.~R{\"{u}}sch, K.~Bleeke, and R.~Kapitza.
\newblock Bloxy: Providing transparent and generic bft-based ordering services
  for blockchains.
\newblock In {\em 38th Symposium on Reliable Distributed Systems, {SRDS} 2019,
  Lyon, France, October 1-4, 2019}, pages 305--314, 2019.

\bibitem{DBLP:journals/corr/abs-2004-10488}
M.~Sigwart, P.~Frauenthaler, C.~Spanring, and S.~Schulte.
\newblock Decentralized cross-blockchain asset transfers.
\newblock {\em CoRR}, abs/2004.10488, 2020.

\bibitem{DBLP:conf/dsn/SousaBV18}
J.~Sousa, A.~Bessani, and M.~Vukolic.
\newblock A byzantine fault-tolerant ordering service for the hyperledger
  fabric blockchain platform.
\newblock In {\em 48th Annual {IEEE/IFIP} International Conference on
  Dependable Systems and Networks, {DSN} 2018, Luxembourg City, Luxembourg,
  June 25-28, 2018}, pages 51--58, 2018.

\bibitem{DBLP:conf/esa/Tamassia03}
R.~Tamassia.
\newblock Authenticated data structures.
\newblock In {\em Algorithms - {ESA} 2003, 11th Annual European Symposium,
  Budapest, Hungary, September 16-19, 2003, Proceedings}, pages 2--5, 2003.

\bibitem{DBLP:conf/fc/TangX0CX18}
Y.~R. Tang, Z.~Xing, C.~Xu, J.~Chen, and J.~Xu.
\newblock Lightweight blockchain logging for data-intensive applications.
\newblock In A.~Zohar, I.~Eyal, V.~Teague, J.~Clark, A.~Bracciali, F.~Pintore,
  and M.~Sala, editors, {\em Financial Cryptography and Data Security - {FC}
  2018 International Workshops, BITCOIN, VOTING, and WTSC, Nieuwpoort,
  Cura{\c{c}}ao, March 2, 2018, Revised Selected Papers}, volume 10958 of {\em
  Lecture Notes in Computer Science}, pages 308--324. Springer, 2018.

\bibitem{Feng_Tian_2016}
F.~Tian.
\newblock An agri-food supply chain traceability system for china based on
  {RFID} {\&} blockchain technology.
\newblock In {\em 2016 13th International Conference on Service Systems and
  Service Management ({ICSSSM})}. {IEEE}, jun 2016.

\bibitem{DBLP:conf/sp/TomescuD17}
A.~Tomescu and S.~Devadas.
\newblock {Catena: Efficient Non-equivocation via Bitcoin}.
\newblock In {\em {SP} 2017}, pages 393--409. {IEEE} Computer Society, 2017.

\bibitem{DBLP:journals/tods/WolfsonJH97}
O.~Wolfson, S.~Jajodia, and Y.~Huang.
\newblock An adaptive data replication algorithm.
\newblock {\em {ACM} Trans. Database Syst.}, 22(2):255--314, 1997.

\bibitem{wood2014ethereum}
G.~Wood.
\newblock Ethereum: A secure decentralised generalised transaction ledger.

\bibitem{DBLP:conf/sp/ZamyatinHLPGK19}
A.~Zamyatin, D.~Harz, J.~Lind, P.~Panayiotou, A.~Gervais, and W.~J.
  Knottenbelt.
\newblock {XCLAIM:} trustless, interoperable, cryptocurrency-backed assets.
\newblock In {\em {IEEE} Symposium on {SP} 2019}, pages 193--210, 2019.

\bibitem{DBLP:conf/icde/Zhang0XTC19}
C.~Zhang, C.~Xu, J.~Xu, Y.~Tang, and B.~Choi.
\newblock Gem{\^{}}2-tree: {A} gas-efficient structure for authenticated range
  queries in blockchain.
\newblock In {\em 35th {IEEE} International Conference on Data Engineering,
  {ICDE} 2019, Macao, China, April 8-11, 2019}, pages 842--853. {IEEE}, 2019.

\bibitem{DBLP:conf/ccs/ZhangCCJS16}
F.~Zhang, E.~Cecchetti, K.~Croman, A.~Juels, and E.~Shi.
\newblock Town crier: An authenticated data feed for smart contracts.
\newblock In {\em {ACM} {SIGSAC} Conference on CCS, 2016}, pages 270--282,
  2016.

\bibitem{DBLP:conf/ccs/ZhangKP15}
Y.~Zhang, J.~Katz, and C.~Papamanthou.
\newblock Integridb: Verifiable {SQL} for outsourced databases.
\newblock In I.~Ray, N.~Li, and C.~Kruegel, editors, {\em Proceedings of the
  22nd {ACM} {SIGSAC} Conference on Computer and Communications Security,
  Denver, CO, USA, October 12-6, 2015}, pages 1480--1491. {ACM}, 2015.

\end{thebibliography}


\begin{thebibliography}{1}

\bibitem{bowman:reasoning}
M.~Bowman, S.~K. Debray, and L.~L. Peterson.
\newblock Reasoning about naming systems.
\newblock {\em ACM Trans. Program. Lang. Syst.}, 15(5):795--825, November 1993.

\bibitem{braams:babel}
J.~Braams.
\newblock Babel, a multilingual style-option system for use with latex's
  standard document styles.
\newblock {\em TUGboat}, 12(2):291--301, June 1991.

\bibitem{clark:pct}
M.~Clark.
\newblock Post congress tristesse.
\newblock In {\em TeX90 Conference Proceedings}, pages 84--89. TeX Users Group,
  March 1991.

\bibitem{herlihy:methodology}
M.~Herlihy.
\newblock A methodology for implementing highly concurrent data objects.
\newblock {\em ACM Trans. Program. Lang. Syst.}, 15(5):745--770, November 1993.

\bibitem{Lamport:LaTeX}
L.~Lamport.
\newblock {\em LaTeX User's Guide and Document Reference Manual}.
\newblock Addison-Wesley Publishing Company, Reading, Massachusetts, 1986.

\bibitem{salas:calculus}
S.~Salas and E.~Hille.
\newblock {\em Calculus: One and Several Variable}.
\newblock John Wiley and Sons, New York, 1978.

\end{thebibliography}
}


\clearpage
\section{Detailed Data-Plane of GRuB}
\subsection{Static Data Structures}
\label{appendix:sec:securitystructures}

\begin{figure*}[!ht]
  \begin{center}
    \subfloat[GRuB's control plane (described in Section~\ref{sec:controlplane}) and data plane (described in Section~\ref{appendix:sec:dataplane})]{%
      \includegraphics[width=0.275\textwidth]{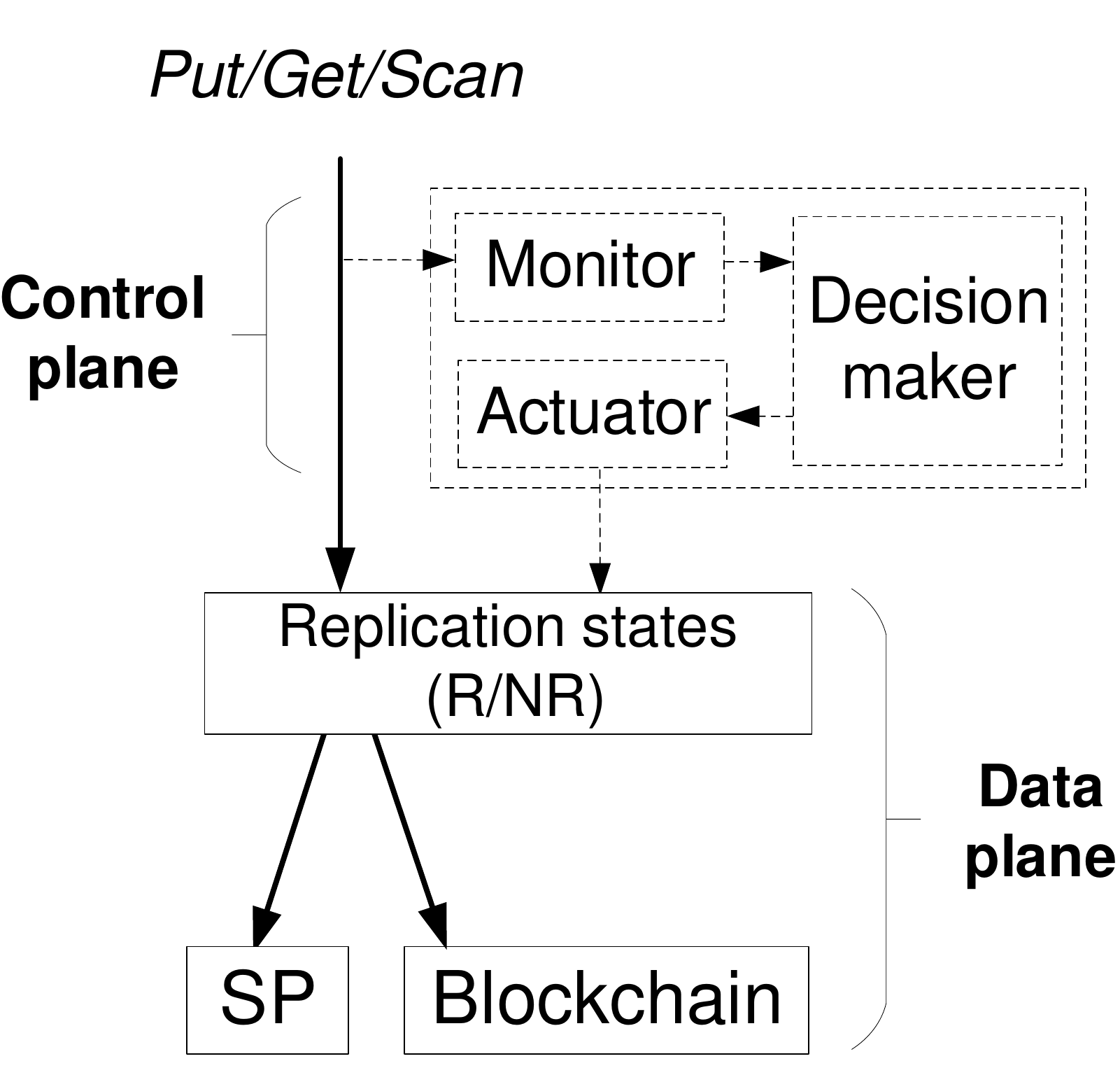}%
      \label{fig:overview}
    }
    \hspace{0.275cm}%
    \subfloat[Detailed view of GRuB system and protocols: Grey boxes are core components implemented in our GRuB prototype]{
      \includegraphics[width=0.275\textwidth]{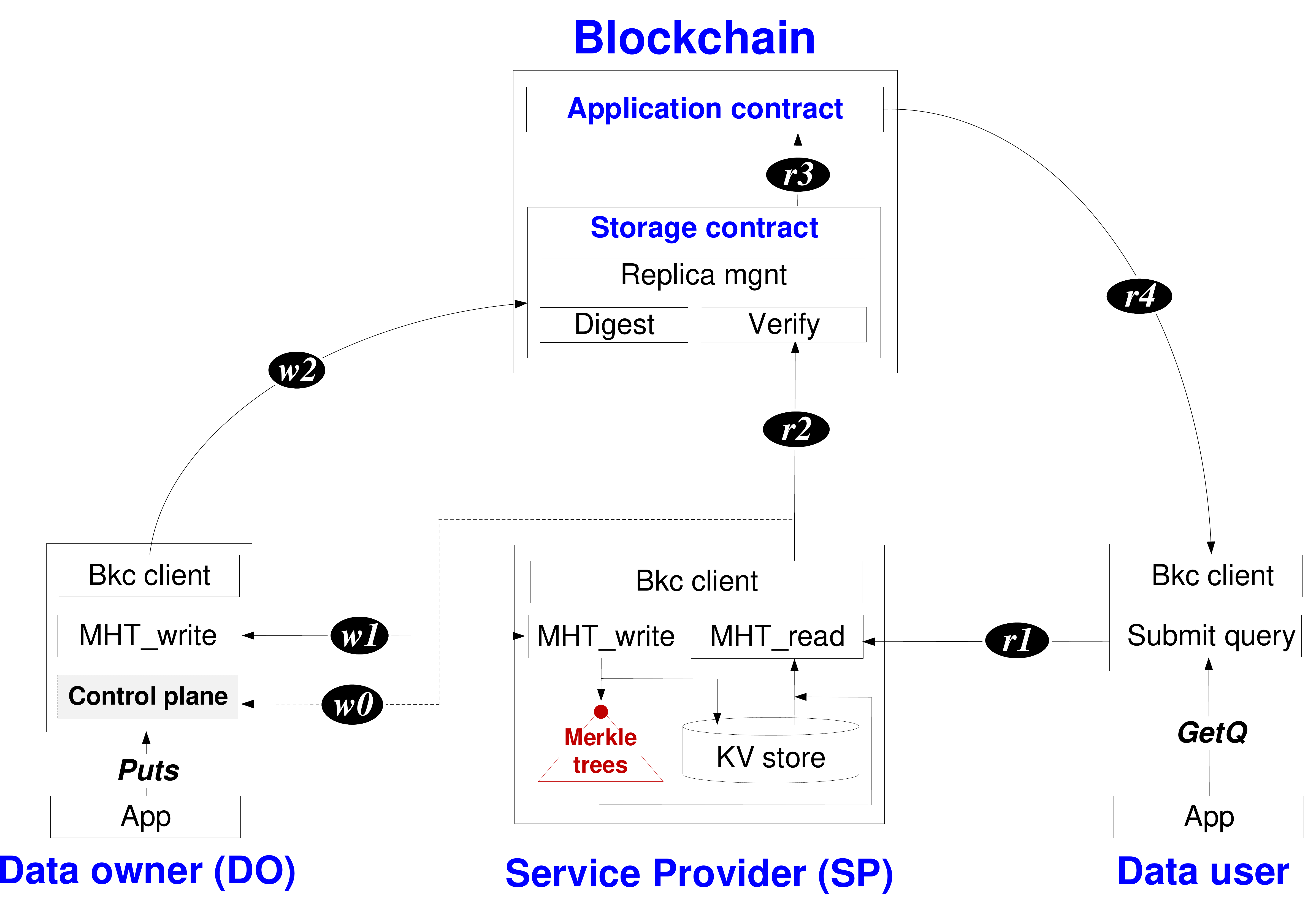}%
\label{fig:protocols}
    }%
  \end{center}
  \caption{GRuB system: data plane, control plane and security data structures}
\label{fig:complex}
\vspace{-0.10in}
\end{figure*}

\label{appendix:sec:staticstructures}
The state of data replication is stored in our system by augmenting the KV records. Each KV record, say $\langle{}k,v\rangle{}$, is associated with a replication state $s$, where the state can be either $NR$ (i.e., not replicated on the Blockchain) or $R$ (i.e., replicated on the Blockchain). 

Recall that the SP maintains a KV store. The KV store on the untrusted SP is protected via an authentication data structure (ADS), which allows for trusted data updates and reads, and prevents the untrusted SP from forging the dataset. For the sake of simplicity, we use a binary Merkle tree to illustrate our system. 
The Merkle tree is built on top of the KV dataset with replication states. Specifically, the data layout on which the Merkle tree is built is the following: KV records are first grouped by the replication states, and then they are sorted by data keys. This data layout allows the Merkle tree to serve the range query on $NR$ data records, which is required in the data-read protocol in GRuB (Section~\ref{appendix:sec:dataplane:reads}). For instance, the four KV records in Figure~\ref{fig:protocols:merkle} fall under two groups, the $NR$ records (i.e., $\langle{}k=w,NR,v=100\rangle{},\langle{}y,NR,200\rangle$) and the $R$ records (i.e., $\langle{}x,R,300\rangle{},\langle{}z,R,400\rangle$). In each group, KV data records are sorted by data keys, and a Merkle tree is built on top of the two groups of KV records, as in Figure~\ref{fig:protocols:merkle}. 
\subsection{Data-Plane Protocols}
\label{appendix:sec:dataplane}

\subsubsection{Data Writes}
\label{appendix:sec:dataplane:writes}

On the write path, DO continuously produces data writes to update the KV dataset. It needs to update both the primary copy stored on SP and the replica stored on the Blockchain. The data writes should be persisted, to one or two data copies, securely against the untrusted SP. To achieve these goals, we design the data-write protocol in GRuB to entail the following steps, which are also illustrated in Figure~\ref{fig:protocols}:

\begin{itemize}
\item[\ballnumber{w1}]
 Given a data write, the DO interacts with SP in an authentication data structure protocol to securely update both $R$ and $NR$ KV records stored on SP.  With a binary Merkle tree, it entails that SP sends a Merkle proof to DO, who updates the Merkle root hash (or digest) locally. Specifically, given a write on KV record $\langle{}k,v\rangle$, the Merkle proof is the Merkle-tree nodes surrounding the path from the leaf of this KV record to the root. An example will be shown in the next paragraph. The DO receiving the proof verifies it against its current root hash and then updates the root hash based on the new KV record $\langle{}k,v\rangle$. 

\item[\ballnumber{w2}] At the end of every epoch, DO batches data writes on replicated KV records. It sends the batched writes and the latest digest in a transaction to the Blockchain. The storage contract on the Blockchain exposes the following contract function:

The storage-management contract processes the write request to update the copy of the digest on the Blockchain and the replicated KV records. There is a delay in updating data replica on the Blockchain which affects the consistency of data reads.
\ifdefined\TTNOTTR 
\else 
\fi

\item[\ballnumber{w0}] 
In the end of every epoch, DO also collects the operation trace and runs a control framework (in Section~\ref{sec:controlplane}) to update data-replication states dynamically. Concretely, DO collects the trace of data writes locally, and the trace of data reads from SP (see Figure~\ref{fig:protocols}). It then feeds the traces as input to the control framework, which produces the output of updated replication states. 

The state transitions are incorporated into the transaction sent in \ballnumber{w2}. The storage-management contract processes state transition $NR\rightarrow{}R$ by inserting the data replica to the buffer in the contract. It processes state transition $R\rightarrow{}NR$ by invalidating an existing replica in the contract buffer. 
\end{itemize}

{\bf Examples}: 
Suppose DO produces one data update in an epoch, $\langle{}w,110\rangle{}$. The initial state of the system, illustrated in Figure~\ref{fig:protocols:merkle}, is that the record is not replicated, $\langle{}w,NR,110\rangle{}$. Assume in Step $\ballnumber{w0}$ the control framework does not update replication state. In Step $\ballnumber{w1}$, the Merkle proof that the DO retrieves from SP is $h_5, h_3$. After verification based on the proof, she then updates the KV record from $\langle{}w,100\rangle$ to $\langle{}w,110\rangle{}$. She also computes the new hash $h_4'=H(\langle{}w,110\rangle{})$ to replace the old $h_4$. The new root hash is calculated $h_1'=H(h_3\|h_2')$ where $h_2'=H(h_4'\|h_5)$. In Step $\ballnumber{w2}$, the transaction sent to the storage contract contains the new root hash ($h_1'$), and nothing else (due to no state transition or no update on replicated records). 

For another example, suppose DO produces one data update $\langle{}x,310\rangle{}$ against the same initial state in Figure~\ref{fig:protocols:merkle}. Here, the original record is replicated, $\langle{}x,R,300\rangle{}$.  Assume in Step $\ballnumber{w0}$ the control framework will update the replication state from $R$ to $NR$ (e.g., by a memoryless decision-making algorithm that is described in Section~\ref{sec:algo:memoryless}). The state transition triggers the KV record to be relocated, that is, leaving the group of $R$ records and joining the group of $NR$ records. In Step $\ballnumber{w1}$, the proof is the Merkle proof of its current location (i.e., $h_2,h_7$) and that of its new location. The new location is between records $\langle{}w,NR,100\rangle{}, \langle{}y,NR,200\rangle$ and the proof of the new location includes $h_3$ and these two records. DO verifies the integrity of the original record using the proof against its current digest. DO then relocates the KV record (due to new replication state) and produces the new root hash. This is done by marking the original KV record invalid (i.e., $h_6'=H(\langle{}x,R,300,invalid\rangle{})$, and by updating the Merkle tree with $h_3'=H(h_6'\|h_7),h_1'=H(h_2\|h_3')$) and by inserting the new KV record at the new location (i.e., $h_8=H(\langle{}x,NR,310\rangle{}),h_9=H(h_4\|h_8),h_2'=H(h_9\|h_5), h_1''=H(h_2'\|h_3')$). In Step $\ballnumber{w2}$, the transaction sent by DO includes the new root hash $h_1''$, the KV record with a state transition, namely $\langle{}x,NR,310\rangle{}$. The storage contract receiving the transaction will invalidate the replica of the original record, $\langle{}x,R,300\rangle{}$.

\subsubsection{Data Reads}
\label{appendix:sec:dataplane:reads}

In the read path, an end blockchain client (not shown in the figure) submits a query to execute application smart contract and read the blockchain state.
The data-read protocol for trusted query processing is described below. The steps are also illustrated in Figure~\ref{fig:protocols}.

\begin{itemize}
\item[\ballnumber{r1}]
The application smart contract, which is data user (DU), issues a \texttt{gGet} request to the storage management smart contract which finds the matching data records (in the case of replication, $R$). If the data is not found, the storage smart contract emit an event to \texttt{request} data from off-chain SP. 

\item[\ballnumber{r2}]
The SP watches the blockchain event log and upon finding a \texttt{request}, it starts to evaluate the query over $NR$ KV records. It retrieves $NR$ KV records that match the query. SP also prepares the query proof from the Merkle tree. The proof consists of the Merkle authentication paths surrounding the $NR$ records. 
\item[\ballnumber{r3}] 
The SP sends the $NR$ records matching the query and the query $q$ itself to the Blockchain. The storage contract receiving the transaction will verify the integrity of $NR$ records and then triggers the query-processing contract to process the query by accessing the retrieved $NR$ records and the $R$ records replicated on the Blockchain. 
\end{itemize}

{\bf Example}: 
Suppose the system initial state is illustrated in Figure~\ref{fig:protocols:merkle} and a data user (DU) smart contract sends a range query $q=[x,z]$ in \ballnumber{r1}. In Step $\ballnumber{r2}$, SP evaluates the range predicate and produces the following $NR$ records that match the query: $\langle{}y,NR,200 \rangle{}$. The Merkle proof is $h_7,\langle{}w,NR,100\rangle{},\langle{}x,R,300\rangle{}$. In Step $\ballnumber{r3}$, the storage-management contract verifies the integrity and completeness of $NR$ records using the proof against the digest stored on the Blockchain.

\section{Additional Evaluation Results}

\subsection{Microbenchmarks}
\begin{figure}[!ht]
  \centering
        \includegraphics[width=0.3\textwidth]{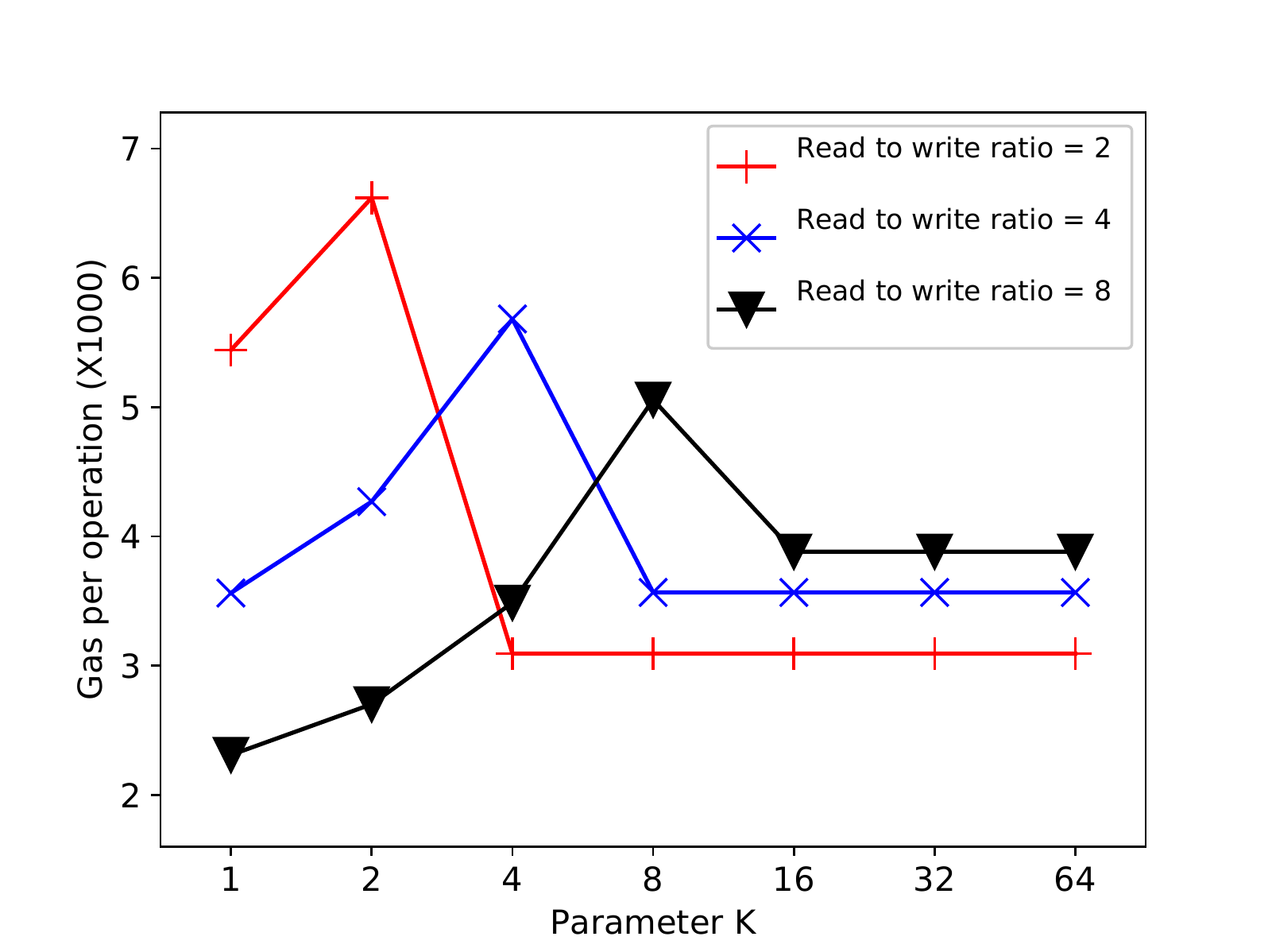}
  \caption{GRuB's Gas with varying parameter $K$}
        \label{fig:mi:k}
\end{figure}

{\bf Varying parameter $K$}:
We additionally measure the Gas of memoryless GRuB with a varying value of $K$. Recall that $K$ states the threshold number of reads to flip the algorithm's decision to replication $R$. In the experiment, we drive workloads of varying read-to-write ratios and report the Gas per operation under different values of $K$. In the experiment result in Figure~\ref{fig:mi:k}, given a fixed workload (say read-to-write ratio being $2$), the Gas first increases with $K$, then decreases and finally stays at a constant value. The highest Gas represents the worst case that the Gas paid for data replication does not result in any Gas savings for future data reads. Before Gas reaches the peak value, increasing $K$ results in increased Gas (due to more off-chain reads and more transactions). After the peak Gas, increasing $K$ results in decreased Gas (due to that off-chain reads are capped in the workload). With different workloads, the value of $K$ under the peak Gas increases along with the read-to-write ratios.

\begin{figure}[!ht]
  \begin{center}
    \subfloat[With varying record size]{%
        \includegraphics[width=0.215\textwidth]{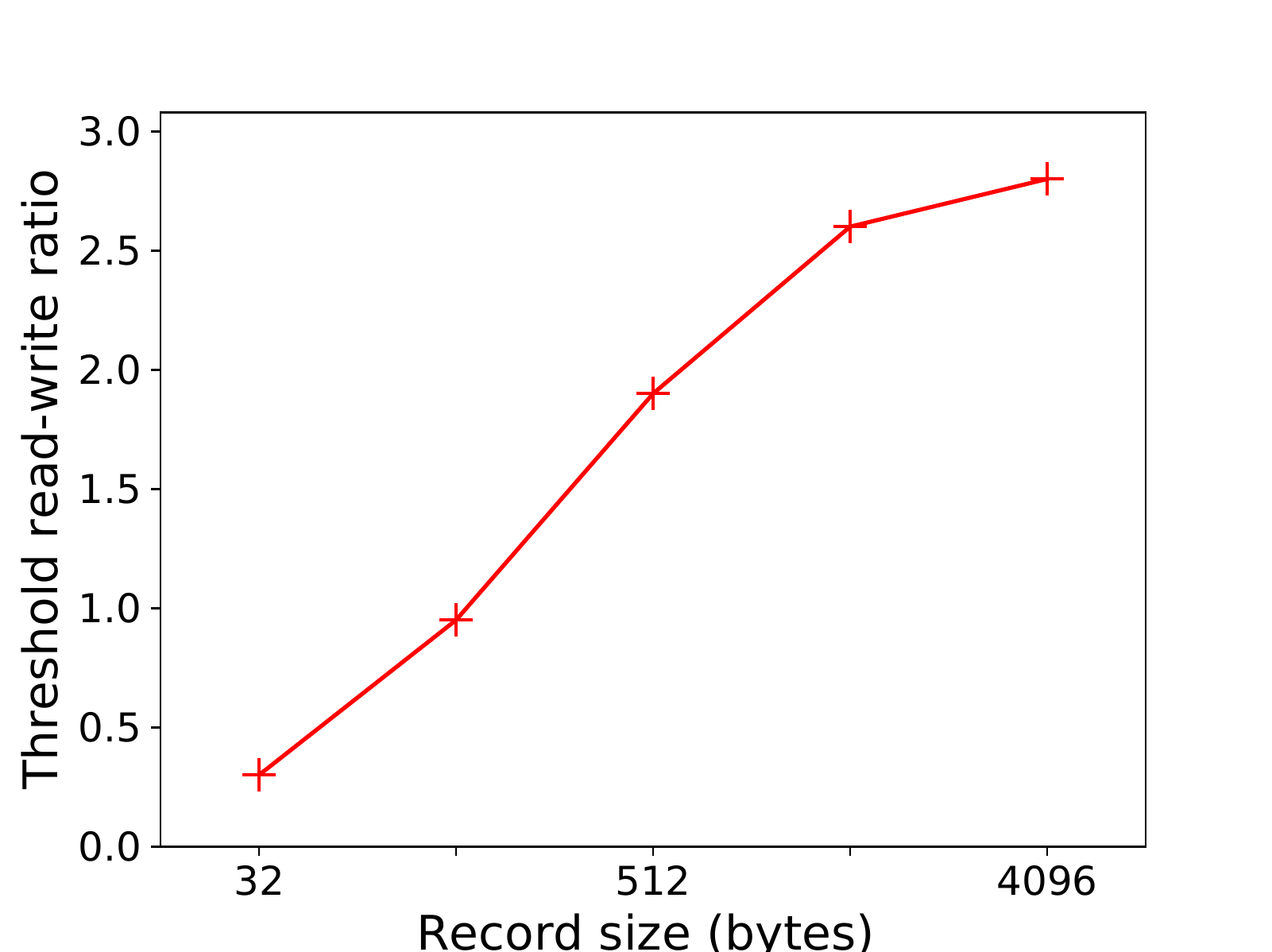}
        \label{fig:mi:cross-point-recordsize}%
    }%
    \hspace{0.2cm}
    \subfloat[with varying data size]{%
        \includegraphics[width=0.215\textwidth]{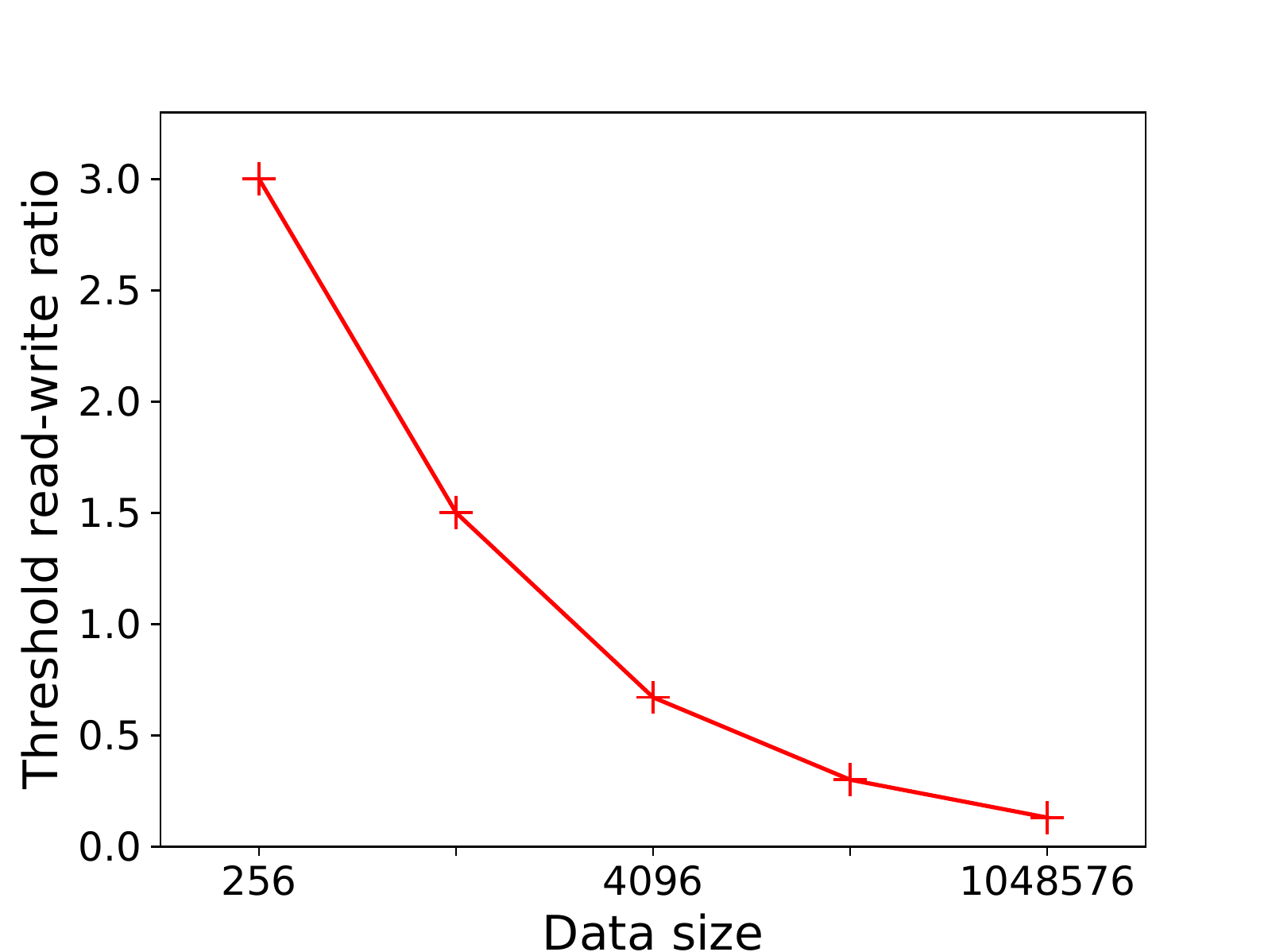}
        \label{fig:mi:cross-point-datasize}%
    }
  \end{center}
  \caption{GRuB's Threshold read-write ratio}
\end{figure}

{\bf Varying threshold read-write ratio}:
The threshold read-write ratio triggers the switch of replication decision and is the case where BL1 and BL2 incur the same Gas.
It affects the range of workloads that dynamic replication schemes can win. We measure the threshold read-write ratio under varying record size and data size. In Figure~\ref{fig:mi:cross-point-recordsize}, as the record size grows, the threshold ratio significantly increases. This is because of the higher unit cost for storage write than that for transactions in the Gas model. In Figure~\ref{fig:mi:cross-point-datasize}, as the data size increases, the proof size grows, which decreases the threshold ratio. This is because the increased proof per record in BL2 needs to be amortized by more writes in BL1. 

\subsection{Macrobenchmarks with YCSB}

\begin{figure}[!ht]
  \begin{center}
    \subfloat[Workload A, E]{%
        \includegraphics[width=0.25\textwidth]{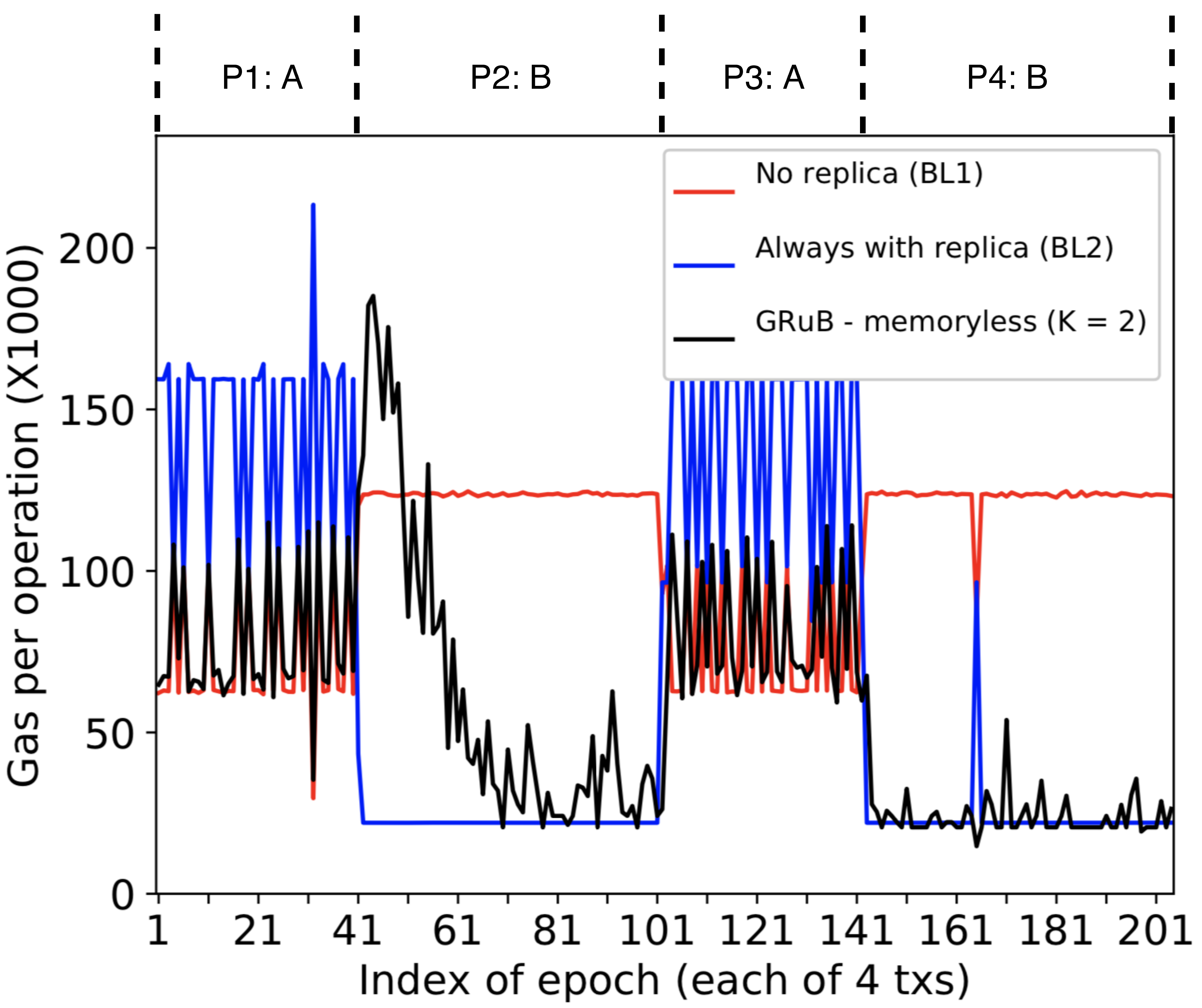}
        \label{fig:ae}
    }%
    \subfloat[Workload A, F]{%
        \includegraphics[width=0.25\textwidth]{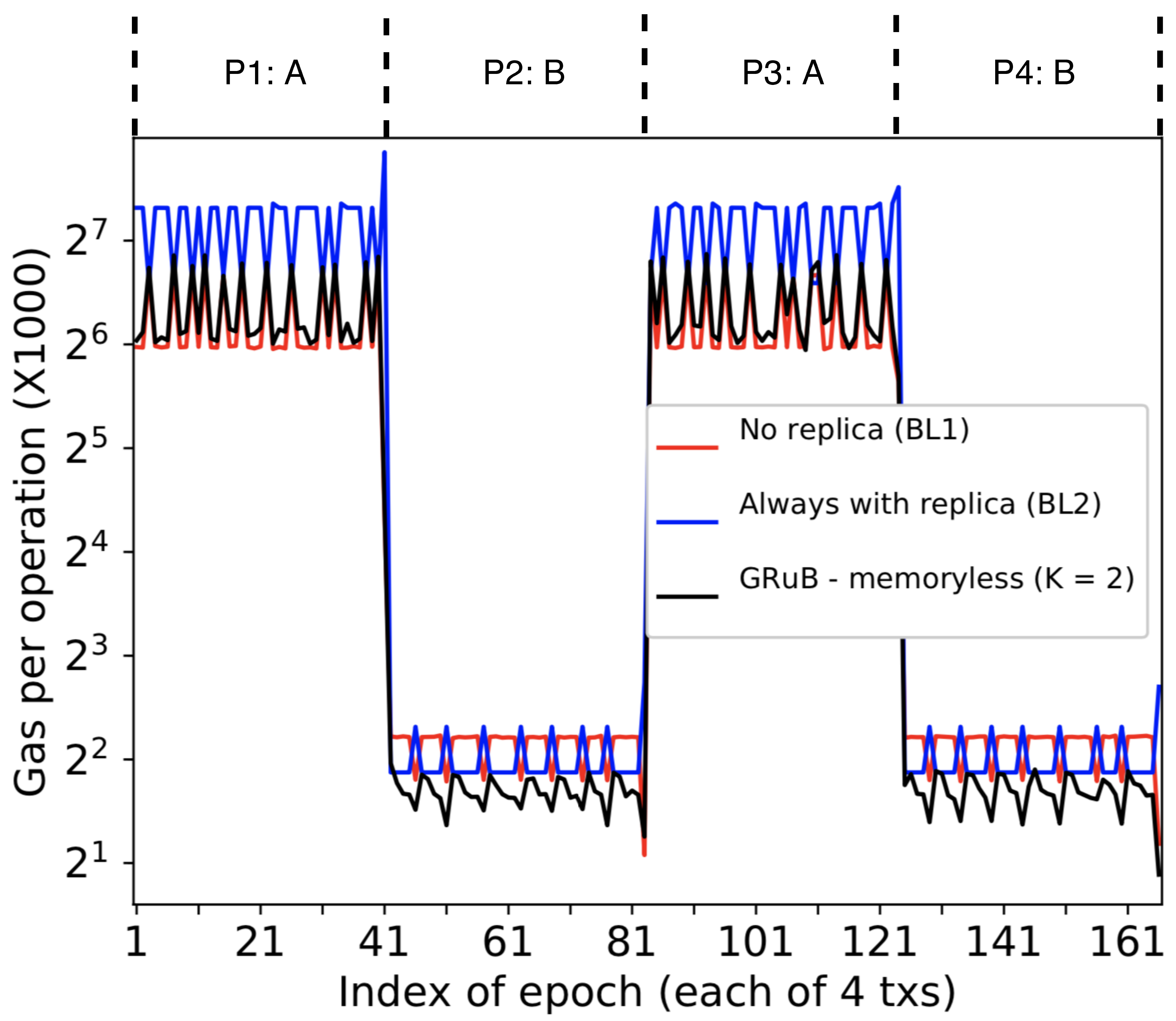}
        \label{fig:af}
    }%
  \end{center}
  \caption{GRuB under mixed YCSB workloads}
        \label{appendix:fig:macro:mixedycsb}
\end{figure}

{\bf Additional YCSB workloads}:
Figure~\ref{fig:ae} shows the Gas per operation under Workload A and E. The results are similar except that the initial cost to replicate KV records is even more significant (notice the Gas spike of GRuB in P2). This can be attributed to that fewer data keys are used in this experiment, which makes a KV record be read multiple times in Phase P2 and P4 and triggers more data replication. The aggregate Gas saving by GRuB reaches $25\%$ against BL1 and $74\%$ against BL2. 

Figure~\ref{fig:af} shows the Gas per operation under Workload A and F. In this experiment, we used small KV records of $32$-byte long. It can be seen that the Gas per operation changes dramatically from Workload A to Workload F. Overall, GRuB saves $54\%$ of the Gas than BL1 and $10\%$ of the Gas than BL2.

\begin{figure}[!ht]
  \begin{center}
     \includegraphics[width=0.35\textwidth]{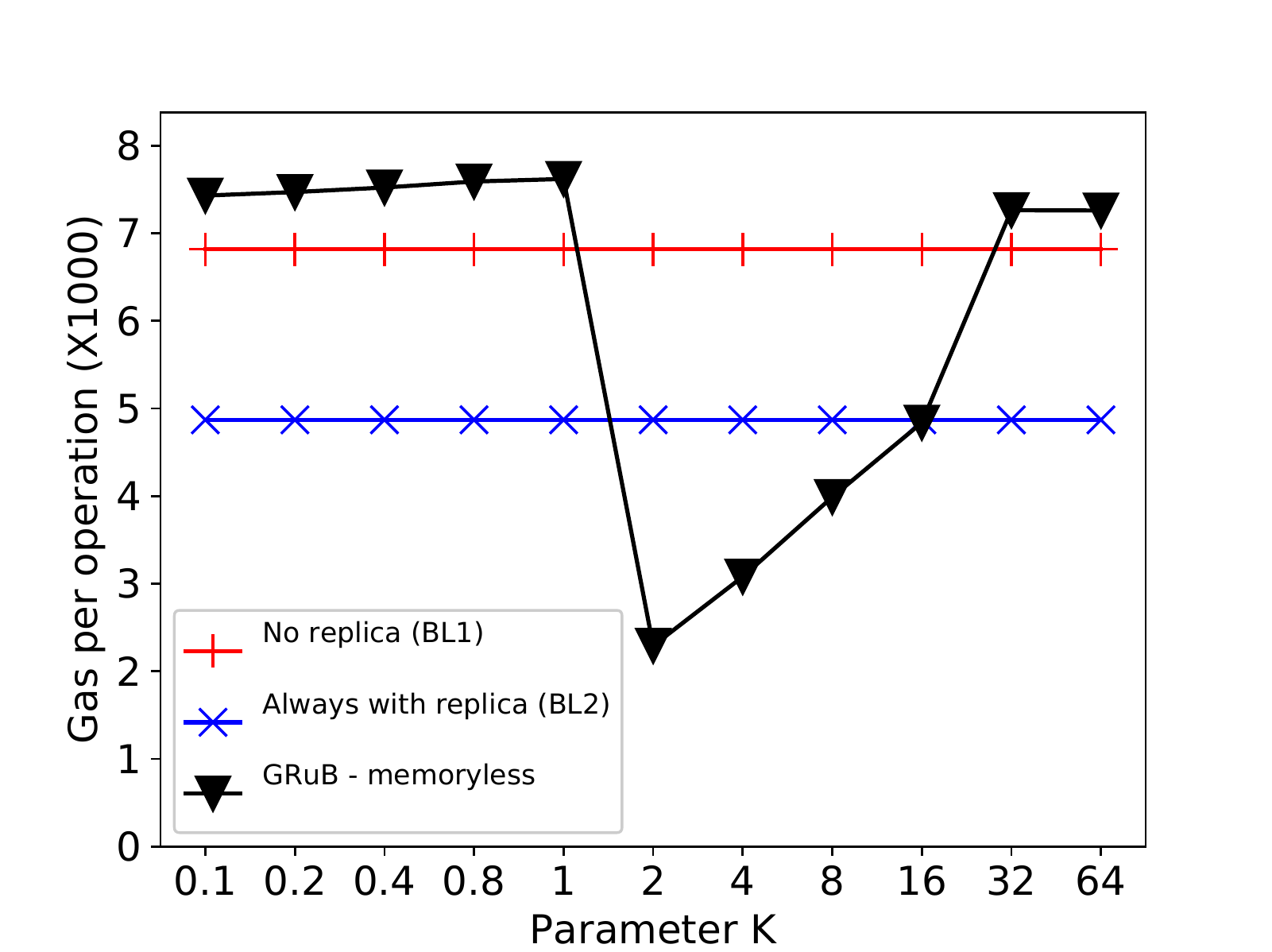}
  \end{center}
  \caption{GRuB's Gas under YCSB with varying K}
  \label{fig:ma:k2}%
\end{figure}

{\bf Varying $K$ under YCSB}: Figure~\ref{fig:ma:k2} illustrates the GRuB's Gas under YCSB workloads, with $K$ varying between $0.1$ to $64$. Compared with the two baselines, GRuB's Gas is sensitive to $K$. When growing $K$, the Gas first decreases, reaches the lowest value (with $K=2$) and then increases. The U shape of the curve can be explained by the following: $K=2$ best matches the tested YCSB workload (in terms of its read-write ratio and its change over time) so that the decision made by GRuB predicts the future workload.

\subsection{Workload-Adaptive $K$}

\begin{figure}[!ht]
  \begin{center}
     \includegraphics[width=0.45\textwidth]{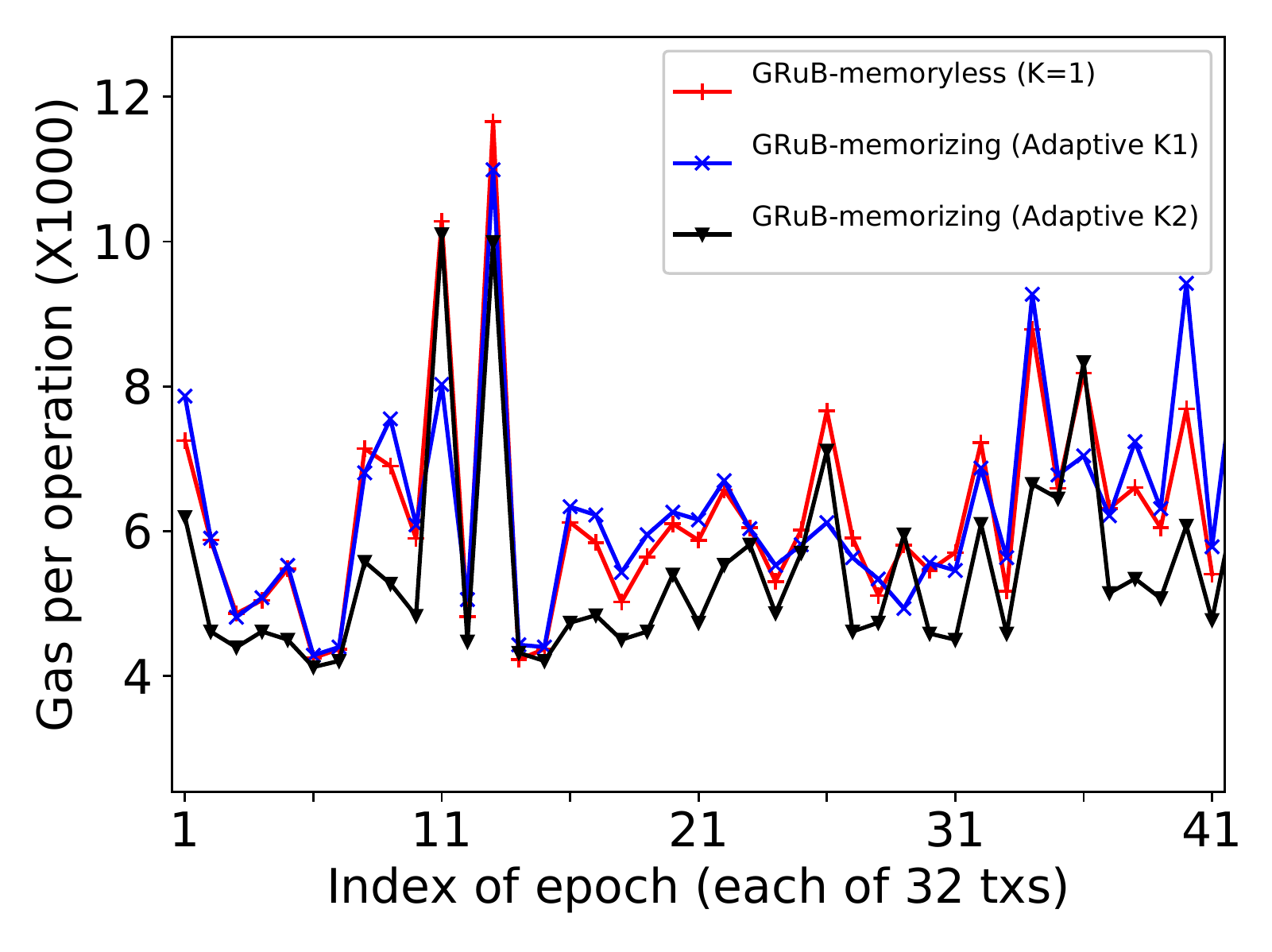}
  \end{center}
  \caption{GRuB with two adaptive-$K$ policies under ethPriceOracle}
  \label{fig:ma:adaptiveK}%
\end{figure}

\begin{table}[!htbp] 
\caption{Aggregated Gas under ethPriceOracle}
\label{tab:adaptiveK}\centering{\small
\begin{tabularx}{0.375\textwidth}{ |X|X| }
  \hline
Approach & Aggregated Gas (X$10^6$) \\ \hline
Memoryless (K=1) & $50.16$ \\ \hline
Memorizing (Adaptive K1) & $50.61 (+0.8\%)$ \\ \hline
Memorizing (Adaptive K2) & $43.74 (-12.8\%)$ \\ \hline
\end{tabularx}
}
\end{table}

{\bf Heuristic for adaptive $K$}: So far, Parameter $K$ in GRuB is set statically and does not change at runtime. Changing $K$ dynamically and adaptively to the current workload may further improve GRuB's cost saving. Here, we present a simple heuristic to adjust $K$ dynamically based on the observation of workload history. In this heuristic, encountering a write, GRuB predicts $K$ to be the average number of reads per write in the selected past trace. As an example setting, we consider the past three writes. If the predicted $K$ is larger than the value of Equation~\ref{eqn:memoryless:parameter}, GRuB decides to replicate the write (i.e., $R$). Otherwise, GRuB decides against replicating the write (i.e., $NR$). We name this policy by Adaptive K1. K1 is based on the intuition that the future will repeat the past. For comparison, we also design a ``dual'' policy, named adaptive K2, whose decision is opposite to K1. That is, when the predicted $K$ is smaller than the value of Equation~\ref{eqn:memoryless:parameter}, GRuB decides against replication ($NR$); otherwise, it decides to replicate ($R$). Thus, K2 reflects that the future does not repeat the past.

{\bf Experiments}: We conduct the experiments using the ethPriceOracle trace (recall Section~\ref{sec:motivateapps} and Section~\ref{sec:cases}) and report the Gas per operation in Figure~\ref{fig:ma:adaptiveK} and Table~\ref{tab:adaptiveK}. The result shows that Adaptive K1 incurs $+0.8\%$ more Gas than the baseline of static K. Adaptive K2 saves about $12.8\%$ Gas than the static K baseline.

{\bf The lesson} we learn here is that the assumption that the future repeats the past does not necessarily hold and is workload specific. In general, the idea of adaptive parameters (e.g., $K$) has the potential of saving the Gas cost further, but to result in actual cost saving, it requires lots of heuristics and workload-specific tuning. Using machine learning techniques to automatically and adaptively find an optimal K is an open research problem and can be addressed in the future work.

\section{Build BtcRelay Benchmarks}

{\bf Building benchmarks: Methodology}
To build the benchmark, we follow the method below: For a Bitcoin-pegged token, we first locate the mint/burn function which has an argument for the Bitcoin depositor/withdraw transaction (Condition M1). Through this, we can bind a token mint/burn operation to a Bitcoin transaction, and further to a Bitcoin block. As a mint/burn operation with on-chain BtcRelay entails reading six Bitcoin blocks, the trace of mint/burn transactions can be converted to a history of Bitcoin-block reads on Ethereum. 
Then, we establish an overall block-read history by combining the reads from different tokens. This block-read history is joined with Bitcoin's native block-write history (i.e., the sequence of Bitcoin blocks produced over time) to establish the block read-write workload. 

We apply the above method to eight Bitcoin-pegged tokens known from etherscan.io~\cite{me:bitcoinpegged}. We found four tokens meet the condition M1, that is, tBTC, imBTC, wBTC and hBTC. Particularly, tBTC has an on-chain BtcRelay and has M1-compatible functions, namely \texttt{provideBTCFundingProof}~\cite{me:tbtc:deposit}/\texttt{provideRedemptionProof}, to receive Bitcoin transactions. The other three actually run BtcRelay off-chain but still expose M1-compatible functions and are amenable for building the benchmarks (e.g., \texttt{addMintRequest}~\cite{me:wbtc:addmintrequest} having an argument storing Bitcoin transaction hash in wBTC). 

Then we collect the mint calls and burn calls from the Ethereum ETL service on Google BigQuery~\cite{me:ethtrace:bigquery}. From there, we combine the four tokens' read traces and join it with the Bitcoin block write sequence~\cite{me:blockchain:info}. 

\begin{figure}
\centering
\subfloat[Number of reads after a write (K means 1,000)]{%
  \includegraphics[width=0.25\textwidth]{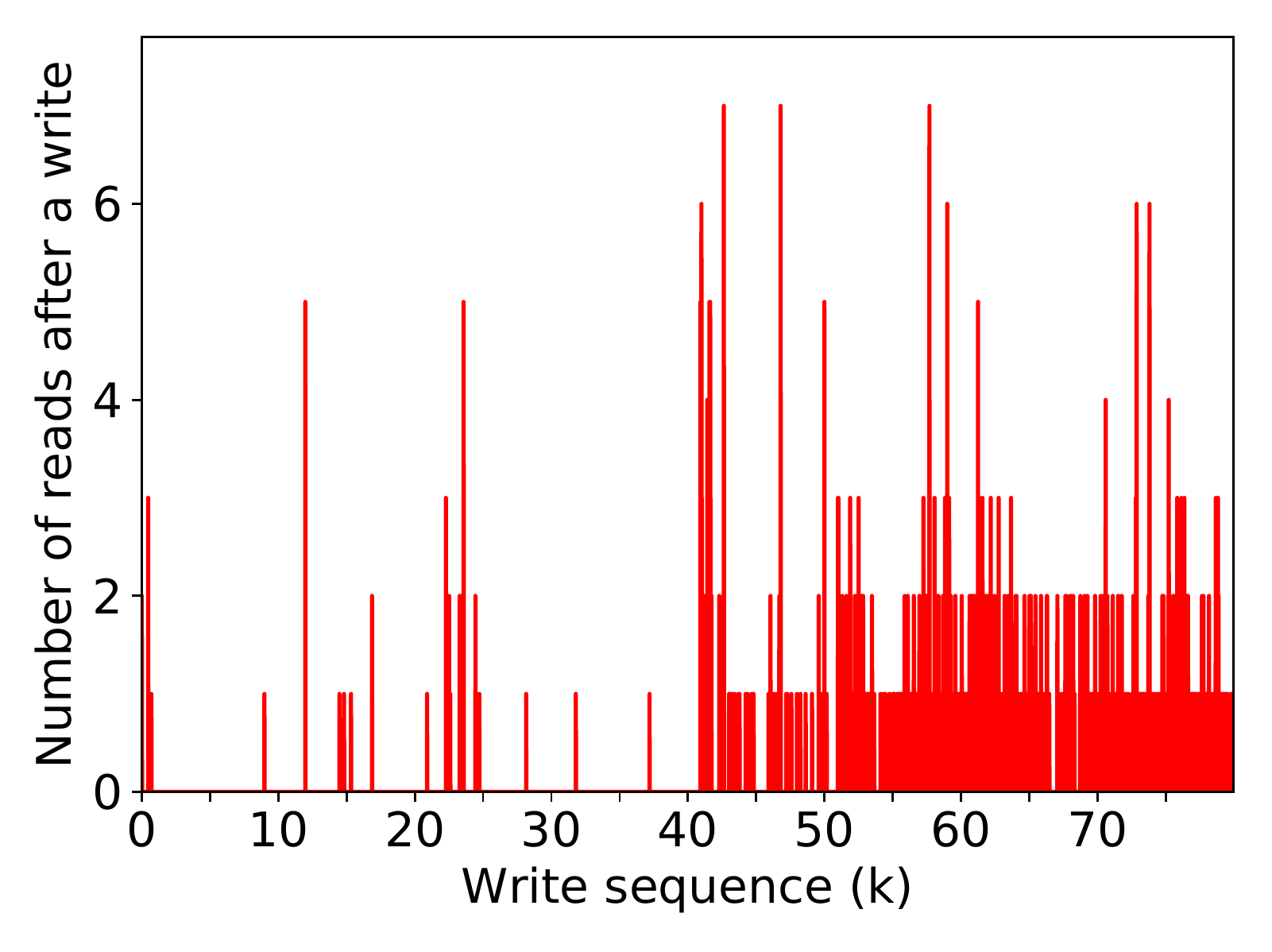}%
\label{fig:workload:btcrelay:1}
}%
\subfloat[Distribution of read-write delay]{%
  \includegraphics[width=0.25\textwidth]{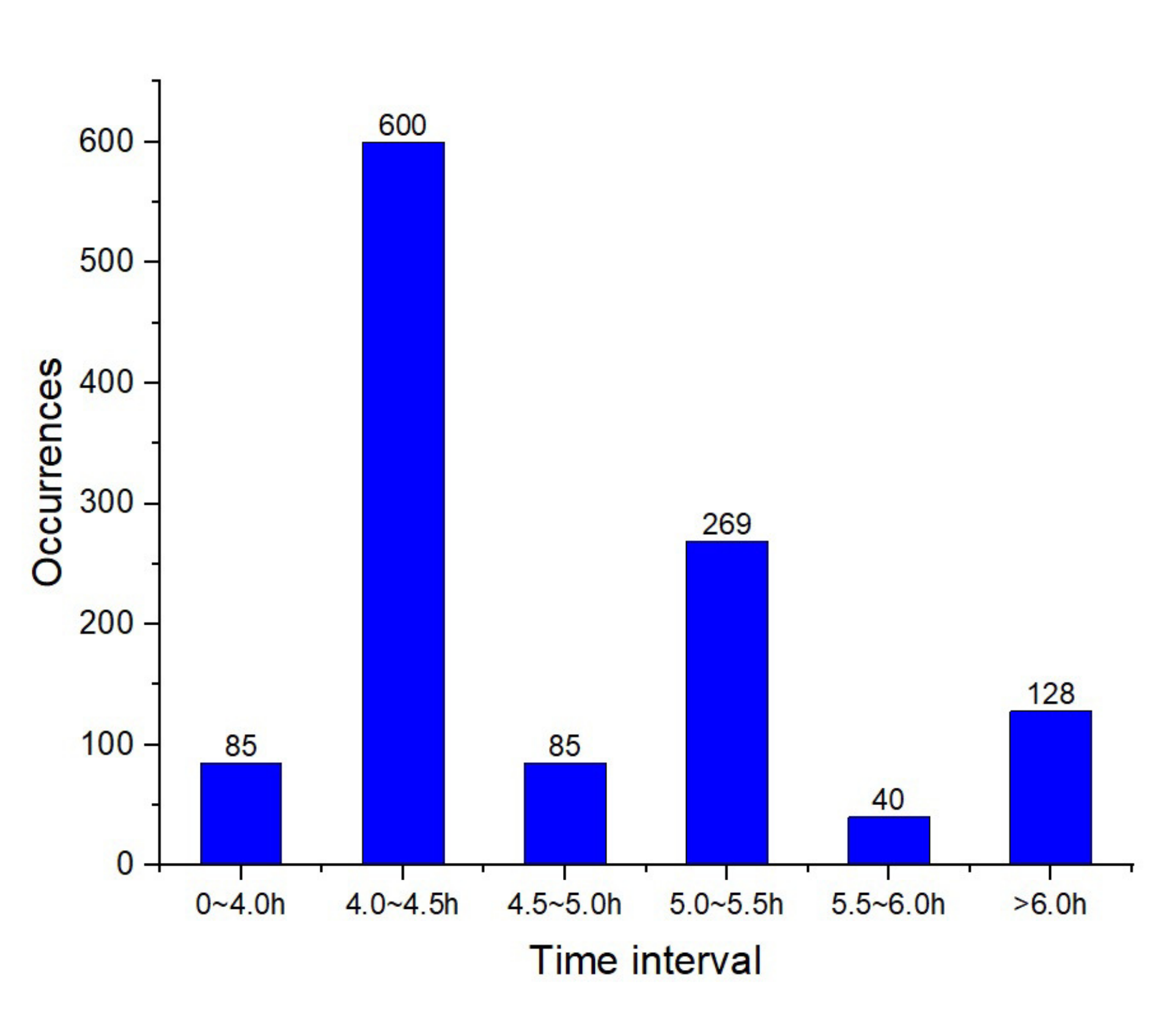}
\label{fig:workload:btcrelay:2}
}%
\caption{The workloads of BtcRelay in existing Bitcoin-pegged ERC20 tokens}
\label{fig:workload:btcrelay}
\end{figure} 

\begin{table}[!htbp] 
\caption{Distribution of writes by the number of reads followed in the BtcRelay trace (\#r represents the number of reads per write)}
\label{tab:btcrelay:distribution}\centering{\small
\begin{tabularx}{0.4\textwidth}{ |X|c|c|c|c|c| }
  \hline
\#r & Percentage & \#r & Percentage & \#r & Percentage\\ \hline
0 & $93.7\%$  & 3 & $0.15\%$ & 6 & $0.02\%$\\ \hline
1 & $5.30\%$  & 4 & $0.05\%$ & 7 & $0.01\%$\\ \hline
2 & $0.77\%$  & 5 & $0.04\%$ & & \\ \hline
\end{tabularx}
}
\end{table}

{\bf Workload analysis}: The workload is presented in Figure~\ref{fig:workload:btcrelay:1}, where we can see the number of reads per write varies between 0 and 7 in the case of BtcRelay. We also measure the ``temporal locality'' of the reads/writes to the same Bitcoin block. In Figure~\ref{fig:workload:btcrelay:2}, most reads of a block occurs 4 hours after the block is produced (written). 

\section{Protocol Consistency}
\label{appdx:properties}
\label{appdx:freshness} 

We first present necessary preliminaries and definitions before analyzing the consistency of the GRuB protocol.

\subsection{Preliminaries and Definitions}

{\bf
Blockchain \& GRuB model}: In a vanilla blockchain, it takes $Pt$ time units to propagate a transaction to all nodes in the blockchain network. It takes an average of $B$ time units to produce a block. Only after $F$ blocks are produced, a transaction is considered finalized in the blockchain network. For instance, 
in Ethereum, $F$ is $250$ and $B$ is $10\sim{}19$ seconds~\cite{wood2014ethereum}.

Consider two concurrent transactions, that is, one transaction is sent after the submission of the other but before the finalization of the other. The ordering of the two transactions is decided non-deterministically by the miners through the underlying consensus protocol of the blockchain~\cite{DBLP:conf/issta/KolluriNSHS19}. Nevertheless, the consensus protocol guarantees the ordering is the same across different nodes eventually after at least one transaction is finalized.

In GRuB, an epoch $E$ is the time interval in which the DO waits and batches data updates in a transaction. 

{\bf
Query-freshness definition}:
Suppose a \texttt{gGet(k)} issued by a DU smart contract at local time $t$ on blockchain node $Ni$ returns a set of KV records $qs$. Query result $qs$ is fresh, w.r.t. delay $d$, if all KV records matching key $k$ and updated on data owner DO before $t-d$ are included in $qs$. Here, it assumes a global clock synchronized across the DO and any blockchain nodes $Ni$. 
Note that query freshness also implies query completeness here.

\subsection{Consistency Analysis}

We focus on analyzing the consistency between a \texttt{gGet(k)} and a \texttt{gPut(k,v)} operation, including read-after-write consistency, operation ordering, et al. 
We consider two general cases: a) The sequential case where \texttt{gGet} is issued sufficiently long after a matching \texttt{gPut}. 
The exact delay $d_0$ between \texttt{gGet} and \texttt{gPut} is $E+P_t+B\cdot{}F$, which will be described later. b) The concurrent case where \texttt{gGet} is issued concurrently (i.e., within delay $d_0$) with a matching \texttt{gPut}. 

{\bf For the concurrent case}, GRuB inherent the non-deterministic but ``eventual'' consistency of the underlying blockchain. That is, whether \texttt{gGet} observes \texttt{gPut}, and more generally, the ordering between \texttt{gGet} and \texttt{gPut}, are non-deterministic but eventually consistency across all blockchain nodes after the finalization of related transactions. Formally, 

\begin{theorem}[Non-deterministic ordering of concurrent gPut/gGet]
Suppose at time $t_1$ the DO submits a \texttt{gPut(k,v)} and at time $t_2$ a blockchain node $Ni$ executes \texttt{gGet(k)}. After $t_2+P_t+B\cdot{}F$, assume the execution of \texttt{gGet(k)} is finalized on the blockchain.

Particularly, when the record \texttt{gGet(k)} accesses is not replicated ($NR$), time $t_2$ refers to when the internal call of \texttt{gGet(k)} is being entered and returned by the blockchain node (the synchronous execution finishes instantly). When the record \texttt{gGet(k)} accesses is not replicated ($NR$), \texttt{gGet(k)} is executed asynchronously and is called back by a \texttt{deliver} transaction. In this case $t_2$ refers to when the \texttt{deliver} transaction is executed on node $Ni$.

If $t_1<t_2<t_1+E+P_t+B\cdot{}F$, \texttt{gPut(k,v)} is said to occur concurrently with \texttt{gGet(k)}. With GRuB, the ordering between concurrent \texttt{gPut(k,v)} and \texttt{gGet(k)} is non-deterministic and is the same across all blockchain nodes after $t_2+P_t+B\cdot{}F$.
\end{theorem} 

We present the justification informally. The non-deterministic ordering in GRuB directly inherits from that of the underlying blockchain. That is, when the record accessed by \texttt{gGet} is replicated ($R$), the ordering between \texttt{gGet} and \texttt{gPut} is reduced to the ordering between the transaction that sends the batch including \texttt{gPut} and the transaction that triggers the DU smart contract to internal-call \texttt{gGet}, which is non-deterministic and eventually consistent due to the blockchain.

When the record accessed by \texttt{gGet} is not replicated ($NR$), the ordering between \texttt{gGet} and \texttt{gPut} is reduced to the ordering between the transaction that sends the batch including \texttt{gPut} and the transaction that calls back the storage smart contract to \texttt{deliver} the read value. The latter ordering between two transactions is non-deterministic and eventually consistent due to the blockchain.

{\bf For the sequential case}, we present the following theorem:

\begin{theorem}[Epoch-bounded query freshness]
GRuB guarantees the \texttt{gGet} query freshness w.r.t. delay $E+Pt+F\cdot{}B$, where the parameters are epoch $E$, block time $B$, propagation delay $Pt$ and the number of blocks needed for finality $F$.
\end{theorem}

\begin{proof}
Consider a \texttt{gGet(k)} query issued at time $t$ by a DU smart contract. The query result includes a KV record $q$. Assume the data update corresponding to $q$ is produced by the DO at time $t'$. Then, the \texttt{gPuts} call that includes $q$ must be issued by the DO at time no later than $t'+E$. \texttt{gPuts} is propagated to all blockchain nodes at the time $t'+E+Pt$. Thus, the time the data update gets finalized in blockchain is  $t'+E+Pt+F\cdot{}B$. If $t\leq{}t'+E+Pt+F\cdot{}B$, the \texttt{gGet} will return record $q$, despite which of the two cases occurs: 1) If $q$ is replicated on blockchain, the \texttt{gGet} directly accesses the on-chain storage and $q$. 

2) If $q$ is not replicated on the blockchain, the \texttt{gGet} retrieves record $q$ from SP via \texttt{request} and \texttt{deliver}. At this time, $q$ is stored on SP's KV store. If the SP is honest, she can and will include $q$ in the \texttt{deliver} call to the calling smart contract. If the SP is dishonest, she may omit the record $q$. If this occurs, GRuB's Merkle trees on SP can ensure that the verification can not pass, as it is hard to forge a non-membership proof on the record that does exist. (Note that the sorted data layout of the KV store on SP is authenticated by the DO).
\end{proof}

\end{document}